\documentclass[11pt]{amsart}
\usepackage[utf8]{inputenc}
\usepackage{amsmath,amssymb,amscd,MnSymbol}
\usepackage{amsthm}
\usepackage{xspace} 
\usepackage[dvips]{epsfig}
\usepackage{graphics}
\usepackage{vaucanson-g}
\usepackage{tikz}
\usetikzlibrary{shapes}
\usetikzlibrary{arrows,positioning,automata}
\tikzstyle{every picture}=[node distance=2cm, inner sep=2pt, on grid, bend angle=20]
\tikzstyle{every state}=[circle, draw, black, fill=white, minimum size=0pt]
\tikzstyle{every edge}=[auto, swap, draw,->,>=stealth']
\tikzstyle{matched edges}=[-, draw, black, dotted, inner sep=0pt]
\tikzstyle{every loop}=[looseness=8]

\usepackage{mathtools}

\usepackage{url}

\newtheorem{proposition}{Proposition}
\newtheorem{corollary}{Corollary}
\newtheorem{theorem}{Theorem}
\newtheorem{lemma}{Lemma}

\theoremstyle{remark}

\newtheorem{example}{Example}

\def\ldots{\mathinner{\ldotp\ldotp}}       

\newcommand{\trace}{\operatorname{trace}}

\newcommand{\bal}{\operatorname{bal}}
\newcommand{\Diag}{\operatorname{Diag}}
\newcommand{\MC}{\operatorname{MC}}
\newcommand{\Fact}{\operatorname{Fact}}
\newcommand{\Dyck}{\operatorname{Dyck}}
\newcommand{\Prime}{\operatorname{Prime}}

\newcommand{\Update}{\operatorname{Update}}
\newcommand{\diff}{\mathop{}\mathopen{}\mathrm{d}}

\def\N{\mathbb{N}}
\def\Z{\mathbb{Z}}
\def\A{\mathcal{A}}
\def\B{\mathcal{B}}
\def\C{\mathcal{C}}
\def\D{\mathcal{D}}

\def\E{\mathcal{E}}
\def\F{\mathcal{F}}
\def\G{\mathcal{G}}
\def\H{\mathcal{H}}

\def\P{\mathcal{P}}

\def\X{\mathsf{X}}

\def\Pgoth{{\mathfrak P}}
\def\MR{\text{MR}}

\newcommand{\ie}{{\itshape i.e.\ }}
\newcommand{\etal}{{\itshape et al.\ }}
\newcommand{\resp}{{resp.\ }}
\ChgStateLineWidth{0.5}
\ChgEdgeLineWidth{0.5}
\FixVCScale{0.4}

\allowdisplaybreaks

\title{Sofic-Dyck shifts}
\author{Marie-Pierre B\'eal}
\address{Universit\'e Paris-Est, Laboratoire d'informatique Gaspard-Monge, UMR 8049 CNRS}
\email{beal@univ-mlv.fr}

\author{Michel Blockelet}
\address{Universit\'e Paris-Est, Laboratoire d'Algorithmique,
  Complexit\'e et Logique}
\email{michel.blockelet@u-pec.fr}

\author{C\v{a}t\v{a}lin Dima}
\address{Universit\'e Paris-Est, Laboratoire d'Algorithmique, Complexit\'e et Logique}
\email{catalin.dima@u-pec.fr}
\thanks{This work is supported by the French National Agency (ANR) through "Programme d'Investissements d'Avenir" (Project ACRONYME $\text{n}^\circ$ANR-10-LABX-58) and through the ANR EQINOCS}

\keywords{Dyck shift, Markov-Dyck shift, sofic-Dyck shift, sofic shift,
  symbolic dynamics, visibly pushdown automaton, visibly pushdown language, zeta function}

\date{\today}
\begin{document}

\begin{abstract}
We define the class of sofic-Dyck shifts which extends the class of
Markov-Dyck shifts introduced by Inoue, Krieger and Matsumoto. 
Sofic-Dyck shifts are shifts of sequences whose
finite factors form unambiguous context-free languages. 
We show that they correspond exactly to the class of shifts of
sequences whose sets of factors are visibly
pushdown languages. We give an expression of the zeta function of a 
sofic-Dyck shift.
\end{abstract}

\maketitle

\section{Introduction}

Shifts of sequences are defined as sets of bi-infinite sequences
of symbols over a finite alphabet avoiding a given set of finite
factors called forbidden factors. Well-known classes of shifts of
sequences are the shifts of finite type which avoid a finite set of
forbidden factors and the sofic shifts which avoid a regular
set of forbidden factors. Sofic shifts may also be defined as
labels of bi-infinite paths of a labeled directed graph.

Dyck shifts are shifts of sequences whose finite factors
are factors of well-parenthesized words. They were introduced by Krieger in \cite{Krieger1974}.
In \cite{Inoue2006}, \cite{KriegerMatsumoto2011}, \cite{InoueKrieger2010}, Inoue,
Krieger, and Matsumoto investigated generalizations of Dyck shifts called Markov-Dyck shifts.
Their languages of factors are unambiguous context-free languages.
Such shifts are presented by a finite-state directed graph equipped
with a graph inverse semigroup. The graph can be considered as an automaton which operates
on words over an alphabet which is partitioned into two disjoint
sets, one for the left parentheses, the
other one  for the right parentheses.
In \cite{InoueKrieger2010}, Inoue and Krieger introduced an extension
of Markov-Dyck shifts by constructing shifts from sofic systems and Dyck
shifts. Examples of shifts of this type are the Motzkin shifts.  
Dyck shifts and their extensions are in general not synchronized but Krieger and Matsumoto
introduced weaker
notions of synchronization suitable
for Markov-Dyck or Motzkin shifts (see \cite{Krieger2006}, \cite{Matsumoto2011b}, \cite{Matsumoto2011c}, 
\cite{KriegerMatsumoto2011b}).
Flow invariants for these shifts are obtained in \cite{Matsumoto2011b}
and \cite{CostaSteinberg2013}. 
In
\cite{Krieger2012} (see also \cite{HamachiKrieger2013} and
\cite{HamachiKrieger2013b}), 
Krieger considers subshift presentations, called 
$\mathcal{R}$-graphs, with
word-labeled edges partitioned into two disjoint sets of positive and negative
edges equipped with a relation $\mathcal{R}$ between 
positive and negative edges going backwards.

In this paper, we introduce a larger class of shifts. We consider
shifts of sequences presented by a finite-state automaton (a
labeled graph) equipped with a set of pairs of edges called matched
edges. The matched
edges may not be consecutive edges of the graph.
We call such structures Dyck automata. 
They may be equipped with a graph semigroup
which is no more an inverse semigroup.
The automaton operates
on words over an alphabet which is partitioned into three disjoint
sets of symbols, the call symbols, the return symbols, and the internal symbols
(for which no matching constraints are required).

We call the shifts presented by Dyck automata sofic-Dyck shifts. We prove that this class is
exactly the class of shifts of sequences whose set of factors is a visibly
pushdown language of finite words. Equivalently, they can be defined
as the sets of sequences which avoid some visibly pushdown language of factors. 
So these shifts could also be called visibly pushdown shifts.

Visibly pushdown languages were introduced by Mehlhorn \cite{Mehlhorn1980} and
Alur \etal \cite{AlurMadhusudan2004, AlurMadhusudan2009}.
They form a natural
and meaningful class inside the
class of unambiguous context-free languages extending the parenthesis languages
\cite{McNaughton1967}, \cite{Knuth1967}, the bracketed languages \cite{GinsburgHarrison1967},
and the balanced languages \cite{BerstelBoasson2002}, \cite{BerstelBoasson2002b}.
These languages share many interesting properties
with regular languages like stability by
intersection and complementation. 
Visibly pushdown languages are used as models for structured data
files like XML files.

We define also a subclass of sofic-Dyck shifts called finite-type-Dyck
shifts. We prove that sofic-Dyck shifts are images of finite-type-Dyck
shifts under proper block maps, \ie block maps mapping call 
(\resp return, internal) symbols to 
call (\resp return, internal) symbols. 
The classes of sofic-Dyck shifts and finite-type-Dyck shifts are invariant by proper conjugacies.

In a second part of the paper, we address the problem of the
computation of the zeta function of sofic-Dyck shift presented by a Dyck automaton.
The zeta function allows to count the
number of periodic points of a subshift. It is a conjugacy invariant of
a class of shifts. Two subshifts which are conjugate (or isomorphic)
have the same zeta functions. The invariant is not complete and it 
is not known, even for shifts of finite type, whether the
conjugacy is a decidable property \cite{LindMarcus1995}.

The formula of the zeta function of a shift of finite type is due to
Bowen and Lanford~\cite{BowenLanford1970}. Formulas for the zeta
function of a sofic shift were obtained by Manning~\cite{Manning1971}
and Bowen~\cite{Bowen1978}. Proofs of Bowen's formula can be found in
\cite{LindMarcus1995} and \cite{Beal1995, Beal1993}. An $\N$-rational
expression of the zeta function of a sofic shift has been obtained
by Reutenauer in~\cite{Reutenauer1997} (see also \cite{BerstelReutenauer1990}).
Formulas for zeta functions of flip systems of finite type are given in
\cite{KimLeePark2003}, and for sofic flip systems in
\cite{KimRyu2011}.
The zeta functions of the Dyck shifts were determined by Keller
in~\cite{Keller1991}.
For the Motzkin shift where some unconstrained symbols are added to the
alphabet of a Dyck shift, the zeta function was determined by
Inoue in~\cite{Inoue2006}. In \cite{KriegerMatsumoto2011}, Krieger and
Matsumoto obtained an expression for the zeta function of a Markov-Dyck shift by
applying a formula of Keller and with a clever encoding of periodic 
points of the shift. 

In Section~\ref{section.zeta}, we give an expression of the zeta
function of a sofic-Dyck shift. The proof
combines techniques used for computing the zeta function of a (non Dyck) sofic
shift and of a Markov-Dyck shift. We implicitly  use the fact that the
intersection of two visibly pushdown languages is a visibly pushdown
language. We give an example of the computation of the zeta function of a sofic-Dyck shift. 

A short version of this
paper appeared in \cite{BealBlockeletDima2014a}.

\section{Shifts} \label{section.shift}

We introduce below some basic notions of symbolic dynamics. We refer
to \cite{LindMarcus1995, Kitchens1998} for an introduction to this theory.
Let $A$ be a finite alphabet. The set of finite sequences or words
over $A$ is denoted by $A^*$
and the set of nonempty finite sequences or words
over $A$ is denoted by $A^+$.
The \emph{shift transformation} $\sigma$ on $A^\Z$ is defined by 
\begin{equation*}
\sigma((x_i)_{i \in \Z}) = (x_{i+1})_{i \in \Z}, 
\end{equation*}
for $(x_i)_{i \in \Z} \in A^\Z$.
A \emph{factor} of a bi-infinite sequence $x$ is a finite word $x_i
\cdots x_j$ for some $i,j$, the factor being the empty word if $j < i$.

A \emph{subshift}  (or \emph{shift}) of $A^\Z$ is a closed
shift-invariant subset of $A^\Z$ equipped with the product of the
discrete topology. If $X$ is a shift, a finite word
is \emph{allowed} for $X$  (or is a \emph{block of} $X$) if it appears as a factor of some bi-infinite
sequence of $X$.  We denote by $\B(X)$ the set of blocks of $X$ and by
$\B_n(X)$ the set blocks of length $n$ of $X$. Let
$F$ be a set of finite words over the alphabet $A$. We denote by
$\X_F$ the set of bi-infinite  sequences of $A^\Z$ avoiding all words of $F$, \ie where no factor belongs to $F$.
The set $\X_F$ is a shift and any shift is the set
of bi-infinite  sequences avoiding all words of some set of finite
words.
When $F$ can be chosen finite (\resp regular), the shift  $\X_F$ is called a \emph{shift of
  finite type} (\resp \emph{sofic}).

Let $L$ be a language of finite words over a finite alphabet $A$. The
language is \emph{extensible} if for any $u \in L$, there are letters
$w,z \in A^+$ such that $wuz \in L$. It is \emph{factorial} if any
factor of a word of the language belongs to the language. 

If $X$ is a subshift, $\B(X)$ is a factorial extensible
language. Conversely, if $L$ is a factorial extensible
language, then the set $\B^{-1}(L)$ of bi-infinite sequences $x$ such
that any finite factor of $x$ belongs to $L$ is a subshift \cite{LindMarcus1995}.


   Let $X \subseteq A^\Z$ be a shift and $m,n$
   be nonnegative integers. A map $\Phi: X \xrightarrow{} B^\Z$ is called
   an $(m,n)$-\emph{block map} with memory $m$ and anticipation $n$ if there
   exists a function $\phi : \mathcal{B}_{m+n+1}(X) \xrightarrow{} B$ such
   that, for all $x \in X$ and any $i \in \Z$, $\Phi(x)_i= \phi(x_{i-m} \dotsm x_{i-1}x_ix_{i+1}
   \dotsm x_{i+n})$. A \emph{block map} is a map which
   is an $(m,n)$-block map for some nonnegative integers $m, n$. 

A \emph{conjugacy} is a bijective block map from $X$ to $Y$.
 A property of subshifts which is invariant by conjugacies 
is called a \emph{conjugacy invariant}. 



\section{Sofic-Dyck shifts} \label{section.soficdyck}

In this section, we define the class of sofic-Dyck shifts which
generalizes the class of Markov-Dyck shifts introduced in \cite{Krieger1974} and \cite{Matsumoto2011} (see also \cite{KriegerMatsumoto2011}).

We consider an alphabet $A$ which is a disjoint union of three finite
sets of letters, the set $A_c$ of \emph{call letters}, the set $A_r$ of
\emph{return letters}, and the set $A_i$ of \emph{internal
  letters}. The set $A=A_c \sqcup A_r \sqcup A_i$ is called a \emph{pushdown alphabet}.

The two sets of call and return symbols may not have the
same size. We assume that any call symbol may match any return symbol.
We denote by $\MR(A)$ the set of all finite words
over $A$ where every return symbol is matched with a call symbol, \ie
$u \in \MR(A)$ if for every prefix $u'$ of $u$, the number of call
symbols of $u'$ is greater than or equal to the number of return symbols
of $u'$. These words are called \emph{matched-return}.
Similarly, $\MC(A)$ denotes the set of all words where 
every call symbol is matched with a return symbol, \ie $u \in \MC(A)$
if for every suffix $u'$ of $u$, the number of return symbols of $u'$ is
greater than or equal to the number of call symbols of $u'$. 
These words are called \emph{matched-call}.
We say that a word is a \emph{Dyck word} if it belongs to the
intersection of $\MC(A)$ and $\MR(A)$. Dyck words are
well-parenthesized or well-formed words. Note that the empty word or
all words over $A_i$ are Dyck words. The set of Dyck words over $A$ is
denoted by $\Dyck(A)$. 
For instance for $A_c =\{ (, [\}$, $A_r  =\{), ]\}$, $A_i = \{i\}$, the word  
$( \: (\: [ \: i  \: )$ is matched-return, the word $( \: ] \: i
\: ]$ is matched-call and $(\:[\:i\:] \:]\:(\:)$ is a Dyck word on $A$.


A \emph{(finite) Dyck automaton} $\A$
over $A$ is a pair $(\G,M)$ of an
automaton (or a directed labeled graph) $\G=(Q,E,A)$ over $A$ where $Q$ is the finite set of
states,
$E \subseteq Q \times A \times Q$ is the set of edges, and with a set $M$ of pairs of edges
$((p,a,q),(r,b,s))$ such that $a \in A_c$ and $b \in A_r$. 
 The edges labeled by call
letters (\resp return, internal) letters are also called \emph{call} (\resp
\emph{return}, \emph{internal}) \emph{edges} and are denoted by $E_c$ 
(\resp $E_r$, $E_i$).
The set $M$ is called the set of \emph{matched edges}.
 If $e$ is an edge we denote by $s(e)$ its
starting state and by $t(e)$ its target state. 


A finite path $\pi$
of $\A$ is said to be an \emph{admissible
  path} if for any factor 
$(p,a,q) \cdot \pi_1\cdot (r,b,s)$ of $\pi$ with $a \in A_c$,
$b \in A_r$ and the label of $\pi_1$ being a Dyck word on $A$, 
$((p,a,q),(r,b,s))$ is a matched pair.  Hence any path of length zero
is admissible
and factors of finite admissible paths are admissible.
A bi-infinite path is \emph{admissible} if all its finite factors are
admissible. 

The \emph{sofic-Dyck shift presented} by $\A$ is the
set of labels of bi-infinite admissible paths of $\A$ and 
$\A$ is called a \emph{presentation} of the shift.

An equivalent semantics of Dyck automata is given in
\cite{BealBlockeletDima2014a} with a graph semigroup associated to
$\A$. This graph semigroup is no more an inverse semigroup as for 
presentations associated to Markov-Dyck shifts \cite{Krieger1974}.

Note that the label of a finite admissible path may not be a block of
the presented shift since a finite admissible path may not be
extensible to a bi-infinite admissible path.

\begin{lemma}\label{lemma.compacity}
The sofic-Dyck shift presented by a  Dyck automaton is exactly the set of bi-infinite
sequences $x$ such that each finite factor of $x$ is the label of a finite
admissible path.
\end{lemma}
\begin{proof}
Let $X$ be the sofic-Dyck shift presented by a Dyck automaton $\A$.
By definition, any finite factor of a bi-infinite sequence of $X$ is
the label of a finite admissible path. 

The converse part is due to the following classical compacity
argument. Let $x$ be a bi-infinite
sequence such that each finite factor of $x$ is the label of a finite
admissible path.
Thus for any positive integer $i$, there is a path
\begin{equation*}
p_{i,-i-1} \xrightarrow{x_{-i}} p_{i,-i}\xrightarrow{x_{i-1}} \dotsm
p_{i,-1}\xrightarrow{x_{0}}p_{i,0}
\xrightarrow{x_{1}}p_{i,1} \dotsm \xrightarrow{x_{i}} p_{i,i},
\end{equation*}
which is admissible for $\A$. For each nonnegative integer $m$, 
there is an infinite number of such paths sharing the states $p_k$ at
all indices $k$ for $-m \leq k \leq m$. 
Then $\pi = ((p_{k-1},x_k,p_k))_{k
  \in \Z}$ is a bi-infinite path whose finite factors are admissible
paths of $\A$.
Thus the label $x$ of $\pi$ belongs to $X$.
\end{proof}

\begin{proposition}
A sofic-Dyck shift is a subshift.
\end{proposition}
\begin{proof}
Let $X$ be a sofic-Dyck shift defined by an automaton $\A$.
Let $F$ be the set of finite words which are not the label of any finite
admissible path of $\A$. Then $X= \X_F$ by Lemma~\ref{lemma.compacity}
and thus $X$ is a subshift.
\end{proof}

We denote respectively by $\MR(X)$, $\MC(X)$ and $\Dyck(X)$, the
intersections of $\MR(A)$, $\MC(A)$ and $\Dyck(A)$ with the set of
blocks of $X$.

\begin{example}
Let $A=A_c \sqcup A_r \sqcup A_i$ with $A_c= \{a_1, \ldots,a_k\}$, $A_r= \{b_1,
\ldots,b_k\}$ and $A_i$ is the empty set.
The \emph{Dyck shift} of order $k$ over the alphabet $A$
is the set of all sequences accepted by the one-state Dyck automaton
$\A=(\G,M)$ containing all loops $(p,a,p)$ for $a \in A$, and where
the edge $(p,a_i,p)$ is matched with the edge $(p,b_i,p)$ for $1\leq i
\leq k$. 

A Motzkin shift is the set of bi-infinite sequences presented
by the automaton $A=A_c \sqcup A_r \sqcup A_i$ with $A_c= \{a_1, \ldots,a_k\}$, $A_r= \{b_1,
\ldots,b_k\}$, the set $A_i$ being no more the empty set.
 A Motzkin shift is represented in the left part of Figure~\ref{figure.example1}.
 It is shown in 
\cite{Inoue2006} that the entropy of the Motzkin shift on this alphabet
is $\log 4$.  Another example is the sofic-Dyck shift $X$ is presented by the Dyck automaton in
the right part of Figure~\ref{figure.example1}. For instance, the bi-infinite sequences
$\cdots ( \: (\: [ \: i \: i \:] \:[\: ]\: ) \cdots$ and $\cdots ) \:
) \: ) \: ) \: ) \cdots$ belong  to $X$ while the sequences
$\cdots (\: [ \: i  \:] \:  [\:  ]\:  ) \cdots$ or $\cdots (\: ]
\cdots$ do not.  We have $(\:i\:i \:)\:(\:) \in \Dyck(X)$ and $( \:)\: (\: [ \:
i \in \MR(X)$.

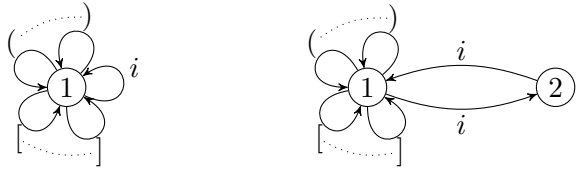
\begin{figure}[htbp]
    \centering
\begin{tikzpicture}[scale=0.4]
\node[state]            (3)         at   (0,0)                 {1};
\path     
(3) edge [in=35,out=335, loop] node {$i$}  (3)
(3) edge [in=105,out=45, loop] node  [inner sep=0mm,pos=0.5]   (a) {)} (3)
(3) edge [in=180,out=120, loop] node  [inner sep=0mm,pos=0.5]   (b) {(} (3)
(3) edge [in=250,out=190, loop] node  [inner sep=0mm,pos=0.5]   (c) {[} (3)
(3) edge [in=330, out=270, loop] node  [inner sep=0mm,pos=0.5]   (d) {]} (3)
(a) edge[matched edges, bend right] node {} (b)
(c) edge[matched edges, bend right] node {} (d);
\node[state]            (1)         at   (10,0)                 {1};
\node[state]            (2) [right =2.5cm of 1]            {2};
\path     
(1)  edge[bend right] node {$i$} (2)
(2)  edge[bend right] node {$i$} (1)
(1) edge [in=100,out=40, loop] node  [inner sep=0mm,pos=0.5]   (a) {)} (1)
(1) edge [in=180,out=120, loop] node  [inner sep=0mm,pos=0.5]   (b) {(} (1)
(1) edge [in=250,out=190, loop] node  [inner sep=0mm,pos=0.5]   (c) {[} (1)
(1) edge [in=330, out=270, loop] node  [inner sep=0mm,pos=0.5]   (d) {]} (1)
(a) edge[matched edges, bend right] node {} (b)
(c) edge[matched edges, bend right] node {} (d);
\end{tikzpicture}
\caption{A Motzkin shift (on the left) over $A=A_c \sqcup A_r \sqcup A_i$ with $A_c= \{(,[\}$, $A_r= \{),]
\}$ and $A_i=\{i\}$. A sofic-Dyck shift (on the right) over the same
tri-partitioned alphabet. Matched edges are linked with a dotted line.
      }\label{figure.example1}
\end{figure}
\end{example}
Note that a call symbol may match several return symbols and conversely although
it is not the case in the above examples.

\section{Finite-type-Dyck shifts} \label{section.finiteTypeDyck}

In this section we give a definition of a subclass of sofic-Dyck
shifts called finite-type-Dyck shifts. We show that sofic-Dyck shifts
are the images of finite-type-Dyck shifts by proper block maps.

Let $A$ and $B$ be two tri-partitioned alphabets. 
We say that a block-map $\Phi:  A^\Z \xrightarrow{} B^\Z$ is
\emph{proper} if and only if $\Phi(x)_i \in A_c$ (\resp $A_r$, $A_i$) 
whenever $x_i \in A_c$ (\resp $A_r$, $A_i$).

Let $A$ be a tri-partitioned alphabet. If $(u,v)$ and $(u',v')$ are two
pairs of words over $A$, we note $(u,v) \preceq (u',v')$ if $u$ is a
suffix of $u'$ and $v$ is a prefix of $v'$.

Let $F \subseteq A^*$ and $U\subseteq   (A^* \times A_c \times A^*) \times  (A^* \times
  A_r \times A^*)$. We say that a finite or bi-infinite sequence $x$
  \emph{avoids} $F$ if, for each finite factor $u$ of $x$, one has $u
  \notin F$.
We say that a finite or bi-infinite sequence $x$
  \emph{avoids} $U$ if for each finite factor $u=vawbz$ of $x$ with $a \in A_c, b \in
A_r$, $w \in \Dyck(A)$, there is no pair $((u_1, a, u_2),(v_1, b, v_2))$ in $U$ such that
$(u_1,u_2) \preceq (v,wbz)$ and $(v_1,v_2) \preceq (vaw,z)$.

A \emph{finite-type-Dyck shift} over $A$ is a set of bi-infinite
sequences $X$ for which there are two \emph{finite} sets $F \subseteq A^*$,
$U\subseteq   (A^* \times A_c \times A^*) \times  (A^* \times
  A_r \times A^*)$,
such that $X$ is the set of sequences \emph{avoiding} $F$ and $U$.

\begin{proposition} \label{proposition.FTDareSD}
A finite-type-Dyck shift is a sofic-Dyck shift.
\end{proposition}
\begin{proof}
Let $X$ be a finite-type-Dyck shift of bi-infinite sequences over $A$
avoiding two finite sets $F$ and $U$. Without loss of generality we
may assume that there are positive integers $m,n$ such that
$F \subseteq A^{m+n+1}$ and $U \subseteq   (A^m \times A_c \times
A^{n}) \times  (A^m \times A_r \times A^{n})$.

We define the Dyck automaton
$\A=(\G,M)$ over $A$ as follows. Let us denote $\G=(Q,E)$. 
We set
\begin{itemize}
\item $Q = \{ (u,v) \mid  u \in A^m, v \in  A^n\}$,
\item $E = \{ ((bu,av), a, (ua,vc)) \mid a,b,c \in A, u \in A^{m-1}, v
  \in A^{n-1}, buavc \notin F\}$,
\item $M$ is the set of pairs of edges  $((du,av), a, (ua,vc))$,
  $((d'u',bv'), b$, $(u'b,v'c'))$, where $a \in A_c$,  $b \in A_r$, $c,c',d,d' \in A$, $u,u' \in A^{m-1}, v,v'
  \in A^{n-1}$ and such that $((du,a,vc),(d'u',b,v'c')) \notin U$.
\end{itemize}
The sofic-Dyck shift presented by $\A$ is $X$.
\end{proof}

\begin{proposition} \label{proposition.SDimageOfFTD}
Sofic-Dyck shifts are the images
of finite-type-Dyck shifts by proper block maps.
\end{proposition}
\begin{proof}
We first show that any sofic-Dyck shift is the image of a
finite-type-Dyck shift by a proper block map.

Let $\A=(\G,M)$ be a Dyck automaton accepting a sofic-Dyck shift $X$
over $A$ with $\G=(Q, E)$. Let $E = E_c \sqcup E_r \sqcup E_i$ be the tri-partitioned alphabet of
edges of $\A$ where $E_c$ (\resp $E_r$, $E_i$) is the set of call (\resp return, internal)
edges of $\A$.

We define a Dyck automaton $\B$ over $E$ as follows. The set of states of
$\B$ is the set of states $Q$ of $\A$. There is an edge $(p,e,q) \in \B$
if and only if $e$ is an edge of $\A$ starting at $p$ and ending in
$q$. A pair of edges $((p,e,q),(r,f,s))$ of $\B$ is matched if $(e,f)$
is a matched pair of $\A$.

Let $Y$ be the sofic-Dyck shift presented by $\B$. It is the set of
sequences avoiding
\begin{itemize}
\item $F = \{ ef \in E^2 \mid t(e) \neq s(f) \}$,
\item $U = \{ ((p,e,q),(r,f,s)) \in  E_c  \times  E_r \mid (e,f)
  \notin M\}$,
\end{itemize}
Since $F$ and $U$ are finite, the shift $Y$ is a finite-type-Dyck shift. 

Let $\Phi: E^\Z \xrightarrow{} A^Z$ be the $(0,0)$-block map defined by
$\phi : \mathcal{B}_1(Y) \xrightarrow{} A$ as follows. We set
$\phi(e) = a$ where $a$ is the label of the edge $e$ of $\A$. 
The map $\Phi$ is clearly a proper block map sending each bi-infinite
admissible path of $\A$ to its label. 
As a consequence $X= \Phi(Y)$.

We now prove that the image of a finite-type-Dyck shift by a proper block map is a
sofic-Dyck shift.

Let $\Phi: A^\Z \xrightarrow{} B^\Z$ be a proper block map and 
$X$ be a finite-type-Dyck shift of sequences over $A$.
Without loss of generality we
may assume that there are positive integers $m,n$ such that 
$\Phi$ is a proper $(m,n)$-block map and
$X$ is the
set of sequences avoiding two finite sets $F$ and $U$ with 
$F \subseteq A^{m+n+1}$ and $U \subseteq   (A^m \times A_c \times
A^{n}) \times  (A^m \times A_r \times A^{n})$. 

Let $\phi: A^{m+n+1} \xrightarrow{} B $ be the function defining $\Phi$.
We define the Dyck automaton
$\A(\phi,F,U)=(\G,M)$ over $A\times B$ as follows. Let us denote $\G=(Q,E)$. 
We set
\begin{itemize}
\item $Q = \{ (u,v) \mid  u \in A^m, v \in  A^n\}$,
\item $E = \{ ((bu,av), (a, \phi(buavc)), (ua,vc)) \mid a,b,c \in A, u \in A^{m-1}, v
  \in A^{n-1} \text{ and } buavc \notin F\}$,
\item $M$ is the set of pairs of edges $(e,f)$ with $e=((bu,av), (a, \phi(buavc)),$
  $(ua,vc))$, 
$f=((b'u',a'v'), (a', \phi(b'u'a'v'c')),
  (u'a',v'c'))$, where $a \in A_c$,  $a' \in A_r$, $b, b', c, c' \in A$,
$u,u' \in A^{m-1}, v,v' \in A^{n-1}$ and the pair $(bu, a, vc), (b'u', a',
v'c') \notin U$.
\end{itemize}
Let $\A_1$ (\resp $\A_2$) be the Dyck automaton obtained by removing
the second (\resp first) components of the labels of the edges. 
The Dyck automaton $\A_1$ is a presentation of $X$. Further, if $x \in
X$, there is a unique admissible path of $\A_1$ labeled by $x$. Indeed, 
each factor of a bi-infinite path of $\A_1$ labeled by $uv$ with $u
\in A^m$,  $v \in  A^n$,
goes through the state $(u,v)$ after reading $u$.
A pair of bi-infinite sequences $(x,y)$ is the label of a bi-infinite
admissible path of $\A(\phi,F,U)$ if and only if $x \in X$ and $\Phi(x)=y$. Hence $\A_2$ is a
presentation of $\Phi(X)$ which is thus sofic-Dyck.
\end{proof}

\begin{proposition} \label{proposition.image}
The image of a sofic-Dyck shift by a proper block map is a
sofic-Dyck shift.
\end{proposition}
\begin{proof}
Let $\Phi$ a proper block map from a sofic-Dyck shift $X$ onto $Y$.
By Proposition \ref{proposition.SDimageOfFTD}, $X$ is the image of a
finite-type-Dyck shift $S$ by a proper block map $\Psi$. The map $\Phi
\circ \Psi: S \xrightarrow{} Y$ is a proper block map and thus its image
$Y$ is a sofic-Dyck shift by Proposition \ref{proposition.SDimageOfFTD}.
\end{proof}

The following corollary is a direct consequence of Proposition \ref{proposition.image}.
\begin{corollary}
The class of sofic-Dyck shifts is invariant by proper conjugacy.
\end{corollary}

We prove below that the same result holds for finite-type-Dyck shifts.

\begin{proposition} \label{proposition.FTDinvariant}
The class of finite-type-Dyck shifts is invariant by proper conjugacy.
\end{proposition}
\begin{proof}
Let $X$ be a finite-type-Dyck shift over $A$ which is properly conjugate to a 
shift $Y$ over $B$. Let $\Phi$ be a proper block map from $A^\Z$ to $B^Z$ 
that induces a conjugacy from $X$ to $Y$.
Without loss of generality we
may assume that there are positive integers $m,n$ such that 
$\Phi$ is a proper $(m,n)$-block map and
$X$ is the
set of sequences avoiding two finite sets $F$ and $U$ with 
$F \subseteq A^{m+n+1}$ and $U \subseteq   (A^m \times A_c \times
A^{n}) \times  (A^m \times A_r \times A^{n})$. 

Let $\Psi= \Phi^{-1}: Y \rightarrow X$ be the proper $(m',n')$-block map
inverse of $\Phi$ and $\psi$ the block function of $\Psi$. 
It induces a map (still denoted by $\psi$) from $B^{m'+m+1+n + n'}$
to $A^{m+1+n}$.
We set $d=m'+m$, $k = n' + n$, $r = d + 1 + k$. 
Let $F' = B^{5r} \setminus \B_{5r}(Y)$ and let $U' \subseteq
(B^d \times B_c \times B^{k}) \times  (B^d \times B_r \times B^{k})$
be the set of pairs $(u',v')$ such that $(\psi(u'),\psi(v')) \in U$.
Let us show that $Y$ is the set $Z$ of sequences avoiding $F'$ and $U'$. 

By construction, $Y \subseteq Z$. Let now $z \in Z$.  We prove by induction that each factor of length $jr$ belong to
$\B(Y)$ for $j \geq 5$. We first have by definition of $Z$ that each factor of length $5r$
belongs to $\B(Y)$. Assume now that each factor of $z$  of length $jr$
belongs to $\B(Y)$ for some $j \geq 5$. 
Let $z'$ be a factor of $z$ of length $2(j-1)r$ decomposed as
$z' = u'_1u'_2 u'_3w' v'_3v'_2v'_1$ with $|u'_1|= |u'_2| = |v'_1| = |v'_2| = r$, $|w'| =
2r$ and $|u'_3| = |v_3| = (j-4)r$. 
The factors $u'=u'_1u'_2 u'_3w'$ and $v'=w' v'_3v'_2v'_1$ of $z'$ are of
length $jr$ and are assumed to be blocks of $Y$.
We set $u = u_1u_2 u_3w_1= \psi(u')$, where $|u_1|=
|u'_1|-m'$ and $|w_1|= |w'|-n'$ and
$v =w_2v_3v_2v_1 = \psi(v')$, where $|w_2|=
|w'|-m'$ and $|v_1|= |v_1|-n'$. 
Note that $w_1[m',|w_1|-1] = w_2[0,|w_2|-n']$. Hence $w_1$ and $w_2$
overlap on a part $w$ of length at least $|w'|-m'-n' =2r-m'-n' \geq m+n$. Since $u', v' \in
\B(Y)$, we have $u, v \in \B(X)$.

Let $\A(\phi,F,U)$ be the Dyck automaton defined in the proof of Proposition
\ref{proposition.SDimageOfFTD}. A pair of bi-infinite sequences $(x,y)$ is the label of a bi-infinite
admissible path of $\A(\phi,F,U)$ if and only if $x \in X$ and
$\Phi(x)=y$. Further, all finite paths of the input Dyck automaton
$\A_1$ of $\A(\phi,F,U)$ which are labeled
by a given block $x_1x_2 \in \B(X)$ with $|x_1| = m$ and
$|x_2| = n$ go through the same state 
after reading
$x_1$. As a consequence, since $|w| \geq m+n$, there is a path in
$\A_1$ labeled by $x'= u_1u_2 u_3w_1tv_3v_2v_1=\psi(z')$,
where $w_2=wt$. 
Since $z'$ avoid $U'$, we have $x'$ avoids $U$. Hence $x' \in \B(X)$,
implying $\phi(x') = u'_0u'_2u'_3w'v'_3v'_2v'_0$,
where $u'_0$ is the suffix of $u'_1$ of length 
$|u'_1|-m'-m$ and $v'_0$ is the prefix of $v'_1$ of length 
$|v'_1|-n'-n$. 
We obtain that $u'_2u'_3w'v'_3v'_2 \in B(Y)$.
Hence each
factor of length $2(j-2)r$ of $z$ belongs to $\B(Y)$ and $2(j-2)r
\geq (j+1)r$ for $j \geq 5$.  This proves that
each factor of $z$ belongs to $\B(Y)$. We get $Z  = Y$ and $Y$ is 
a finite-type-Dyck shift.
\end{proof}




\section{Presentations of sofic-Dyck shifts}

In this section we define several particular presentations of
sofic-Dyck shifts which will be useful for the computation of zeta function.

A Dyck automaton is \emph{deterministic}\footnote{Deterministic
  presentations are also called \emph{right-resolving} in \cite{LindMarcus1995}.} if
there is at most one edge starting in a given state and with a given
label. Sofic shifts (see \cite{LindMarcus1995}) always have a deterministic
presentation. Although visibly pushdown languages are accepted by deterministic
visibly pushdown automata \cite{AlurMadhusudan2009}, sofic-Dyck shifts
may not be presented by any deterministic Dyck automaton as is shown
in Example \ref{example.example2}.
Indeed, the two notions of determinism do not match. The notion of
 determinism for visibly pushdown languages includes the stack
symbol as input for return transitions of visibly pushdown automata.


Let $\A$ be a Dyck automaton. We define the \emph{left reduction} of
$\A$ as the Dyck automaton obtained through some determinization
process. The process is an adaptation to Dyck automata of the determinization 
of visibly pushdown automata \cite{AlurMadhusudan2004}. It is
sketched in \cite{BealBlockeletDima2014b} and we detail it here.

Let $\A=(\G,M)$ with $\G=(Q,E)$ be a Dyck automaton over $A$.
We define a Dyck automaton $\D=(\H,N)$ over $A$, where $\H=(Q',E')$
with $Q'= \Pgoth(Q \times Q)
\times \Pgoth(Q)$ and $\Pgoth(Q)$ is the
set of subsets of $Q$. States are pairs $(S,R)$ where $S$ is called
the \emph{summary}\footnote{The definition of summaries differs
  slightly from the one given in \cite{AlurMadhusudan2004}.}
of the state and $R$ is a nonempty subset of $Q$. 
The state $I=(\emptyset, Q)$ is called the initial state.
For each state $(S,R)$, the set $S$ is empty if and only 
the admissible paths going from $I$ to $(S,R)$ are labeled by a matched-call word.
It is nonempty if all admissible paths going from $I$ to $(S,R)$ are of the form
\begin{equation*}
I \xrightarrow{u} (S",R") \xrightarrow{a} (T,U) \xrightarrow{w} (S,R),
\end{equation*}
where $a \in A_c$ and $w$ is a Dyck word.
If there is such a path, the summary $S$ of the state $(S,R)$ is the set of
pairs $(p,q)$ in $U \times R$ such that there is an admissible path of $\A$
labeled by the Dyck word $w$ from $p$ to $q$. In both cases, if there is a
path labeled by $v$ in $\D$ from $I$ to $(S,R)$, then $R$ is the set
of states $q$ such that there is an admissible path in $\A$ labeled
by $v$ ending in~$q$.

For a subset $R$ of $Q$, we  denote by $\Diag(R)$ the set of
all pairs $(p,p)$ for $p \in R$. 
The edges of $\D$ are defined as follows.

\begin{itemize}
\item For every $\ell \in A_i$, $((S,R),\ell,(S',R')) \in E'$ if
  $S'=\{(p,q) \mid \exists r \in Q,  (p,r) \in S, (r,\ell,q) \in E\}$
and 
$R'=\{q \mid \exists p \in R, (p,\ell,q) \in E\}$ is nonempty.
\item For every $a \in A_c$, $((S,R), a, (\Diag(R'),R')) \in E'$ if
$R'=\{q \mid \exists p \in R, (p,a,q) \in E\}$ is nonempty.
\item For every $b \in A_r$, the edges stating from $(S,R)$ with $S
  \neq \emptyset$ labeled by $b$ are defined as follows. 
For any edge $((S'',R''), a, (T,U))$ with $a \in A_c$ we define
\begin{itemize}
\item $\Update = \{(p,p') \mid \exists p_1,p_2 \colon (p,a,p_1) \in E,
(p_1,p_2) \in S, (p_2,b,p') \in E, ((p,a,p_1), (p_2,b,p')) \in M\}$,
\item $S'=\{(p,q) \mid \exists p', (p,p') \in S", (p',q) \in
\Update\}$,
\item $R'=\{q \mid \exists p \in R", (p,q) \in \Update \}$.
\end{itemize}
If $R'$ is not empty, we define an edge $((S,R), b, (S',R')) \in E'$
and set this edge matched with $((S",R"), a, (T,U))$.
\item For every $b \in A_r$, we define an edge $((\emptyset,R), b, (\emptyset,V)) \in E'$
where 
$V=\{q \mid \exists p \in R, (p,b,q) \in E\}$ is nonempty. This
return edge is not
matched with any call edge.
\end{itemize}
We only keep in $\D$ the states reachable from $I$.

\begin{proposition}
The left reduction of a Dyck automaton $\A$ presents the same
sofic-Dyck shift as $\A$.
\end{proposition}
\begin{proof}
Let $X$ be the sofic-Dyck presented by $\A$ and $\D$ be the left
reduction of $\A$.
Let $v$ be the label of an admissible path of $\A$ going from $p$ to
$q$. By construction there is an admissible path of $\D$ labeled by $v$
going from $I$ to some state $(S,R)$ with $q \in R$. 
Thus labels of finite admissible paths of $\A$ are labels of finite 
admissible paths of $\D$. 

Conversely, let $v$ be the label of some finite admissible path $\pi$ of
$\D$. We claim that $v$ is the label of an admissible path of $\A$.

We prove the claim by recurrence on the length of $v$.  It is 
true if $v$ is the empty word.
Let $v = uc$ where $c \in A$ and 
$\pi = (S_1,R_1) \xrightarrow{u} (S,R) \xrightarrow{c} (S',R')$.
By induction hypothesis, we assume that for any state $r \in R$
there is an admissible path labeled
by $u$ from some state $q\in R_1$ to $r$.
If the edge $((S,R),c,(S',R'))$ is a call or internal edge or is a
return edge not matched with a call edge of $\pi$, the result holds 
by construction for $uc$.
Let us assume that 
\begin{equation*}
\pi = (S_1,R_1) \xrightarrow{u} (S",R") \xrightarrow{a} (T,U) \xrightarrow{w} (S,R)
\xrightarrow{b} (S',R'),
\end{equation*}
where $w$ is a Dyck word over $A$, $a \in A_c$, and $b \in A_r$.
By induction hypothesis, we assume that for any state $p \in R"$
there is an admissible path $q \xrightarrow{u} p$
in $\A$ for some $q \in R_1$.

For any $r \in R'$ there are $ p \in R''$ and $(p_1,p_2) \in S$ such
that $(p, a, p_1)$ and $(p_2, b, r)$ are matched in $\A$. Further, $S$ is the set
of pairs $(s,s') \in U\times R$ such that there is an admissible path
in $\A$ labeled by
$w$ from $s$ to $s'$.  It follows that there is in $\A$ an admissible
path labeled by $w$ from $p_1$ to $p_2$ and thus an admissible path
$q \xrightarrow{u} p \xrightarrow{a} p_1 \xrightarrow{w} p_2 \xrightarrow{b} r$ in $\A$
which concludes the proof of the claim.
\end{proof}
Note that since the label of an admissible path of $\D$ is the label
of an admissible path of $\A$, each label of an admissible path of
$\D$ is the label of an admissible path of $\D$ starting at $I$. 

We similarly define the
\emph{right reduction} of $\A$ with a co-determinization of 
$\A$ and an exchange of roles played by call and return
edges. Note the left reduction of $\A$ may have more states
than $\A$.

Let $L$ be a language of finite words. A Dyck automaton is
\emph{$L$-deterministic} if there is at most one admissible path starting in a
given state and with a given label in $L$.

By construction the left reduction of a Dyck automaton is
$A_c$-deterministic and $A_i$-deterministic.

A Dyck automaton is \emph{weak-deterministic} if there is a state $I$
such that for any word $u$ there is at most one admissible path
labeled by $u$ starting at $I$.

\begin{proposition}\label{proposition.weakDeterministic}
The left reduction of a Dyck automaton is weak-deterministic. 
\end{proposition}
\begin{proof}
Let $\D$ be the left reduction of a Dyck automaton $\A$ and let $I$ be
the initial state of $\D$. Let us suppose that the property is false.
We consider two minimal-length distinct admissible paths starting at $I$ and
sharing the same label. 
\begin{align*}
I &\xrightarrow{u} (S,R) \xrightarrow{b} (T,U), \\
I &\xrightarrow{u} (S',R') \xrightarrow{b} (T',U'), \\
\end{align*}
with $b \in A$ and $(T,U) \neq (T',U')$.
We may assume $(S,R) =(S',R')$ since these
paths are of minimal length. Since $\D$ is $A_c$-deterministic and
$A_i$-deterministic, we may assume that $b \in A_r$.
By definition, we have $U=U'$.
If $ub$ is matched-call, $T$ and $T'$ are empty, hence
$(T,U) = (T',U')$.
If $ub = u'awb$, where $w$ is a Dyck word and $a \in A_c$, the two above paths are
\begin{align*}
I &\xrightarrow{u'} (S_1,R_1) \xrightarrow{a}  (S_2,R_2)  \xrightarrow{w} (S,R) \xrightarrow{b} (T,U), \\
I &\xrightarrow{u'} (S_1,R_1)  \xrightarrow{a} (S_2,R_2)
\xrightarrow{w}  (S,R) \xrightarrow{b} (T',U). \\
\end{align*}
Since the paths are admissible, $((S,R),b,(T,U))$ is matched with
$((S_1,R_1), a,$ $(S_2,R_2))$ and $((S,R),b,(T',U))$ is matched with
$((S_1,R_1),a, (S_2,R_2))$. By definition of the summary we get $T =
T'$, a contradiction.
\end{proof}

\begin{corollary} \label{proposition.DyckDeterministic}
The left reduction of a Dyck automaton over $A$ is $Dyck(A)$-deterministic. 
\end{corollary}
\begin{proof}
Let $(S,R)$ be a state of the left reduction of a Dyck automaton over
$A$ and $w$ be a Dyck word over $A$. Let us assume that there are two 
admissible paths $\pi_1$ and $\pi_2$ labeled by $w$ starting at $(S,R)$. 
Since there is an admissible path $\pi$ from $I$ to $(S,R)$, the paths 
$\pi\pi_1$ and $\pi \pi_2$ are two admissible paths starting at $I$. 
They are then equal by Proposition \ref{proposition.weakDeterministic}.
\end{proof}
\begin{example}
The Dyck automaton $\A$ on the left of Figure
\ref{figure.exampleLetReduction} has as left reduction the Dyck
automaton on the right of the picture.  The initial state is the state $I=(\emptyset,\{1,2,3\})$.
\begin{figure}[htbp]
    \centering
\tikzstyle{every edge}=[auto,draw,->,>=stealth'] 
\tikzstyle{every loop}=[looseness=8]
\tikzset{packet/.style={rectangle, draw,
minimum size=6mm, rounded corners=1mm}}
\begin{tikzpicture}[scale=0.4]
\node[packet]            (1a)         at   (21,0) {(1,1);1};
\node[packet]            (123)       at   (10,0) {$\emptyset$;123};
\node[packet]            (2b)         at   (17,2.8) {$\emptyset$;2};
\node[packet]            (2a)         at   (25,3) {(1,2);2};
\node[packet]            (3b)         at   (25,-3) {(1,3);3};
\node[packet]            (3a)         at   (17,-2.8) {$\emptyset$;3};
\node[packet]            (1b)         at   (13,-0) {$\emptyset$;1};
\node[state]            (1)         at   (-3,0)                 {1};
\node[state]            (2)         at    (1,3)             {2};
\node[state]            (3)         at    (1,-3)            {3};
\path     
(1)  edge[bend left] node [inner sep=0.5mm,pos=0.5] (1b2) {$b$} (2)
(2)  edge[bend left] node {$i$} (1)
(1)  edge[bend right,swap] node [inner sep=0.5mm,pos=0.5] (1b3) {$b$} (3)
(3)  edge[bend right,swap] node {$i$} (1)
(2) edge [in=35,out=335, loop,swap] node {$j$}  (2)
(3) edge [in=35,out=335, loop,swap] node {$k$}  (3)
(1) edge [in=160,out=100, loop,swap] node  [inner sep=0.5mm,pos=0.5]   (loopa) {$a$} (1)
(1) edge [in=260,out=200, loop,swap] node  [inner sep=0.5mm,pos=0.5]   (loopaprime)
{$a'$} (1)
(loopa) edge[matched edges, bend left] node {} (1b2)
(loopaprime) edge[matched edges, bend right] node {} (1b3)
(123)  edge[bend left=40,swap,pos=0.3] node {$j$} (2b)
(123)  edge[bend right=40,pos=0.3] node {$k$} (3a)
(123)  edge[bend right=80,swap] node [inner sep=0.5mm,pos=0.7] (123-a'-1a) {$a'$} (1a)
(123)  edge[bend left=80] node [inner sep=0.5mm,pos=0.7] (123-a-1a) {$a$} (1a)
(1a)  edge[bend left=0,swap] node [inner sep=0.5mm,pos=0.5] (1a-b-2b) {$b$} (2b)
(1a)  edge[bend right=0] node [inner sep=0.5mm,pos=0.6] (1a-b-3a) {$b$} (3a)
(1a)  edge[bend left=10] node [inner sep=0.5mm,pos=0.6] (1a-b-2a) {$b$} (2a)
(2a)  edge[bend left=10] node [inner sep=0.5mm] {$i$} (1a)
(1a)  edge[bend right=10,swap] node [inner sep=0.5mm,pos=0.6] (1a-b-3b) {$b$} (3b)
(3b)  edge[bend right=10,swap] node [inner sep=0.5mm] {$i$} (1a)
(1a) edge [in=90,out=60, loop,swap] node  [inner sep=0.5mm,pos=0.55]   (bouclea){$a$} (1a)
(1a) edge [in=-60,out=-90, loop,swap] node  [inner sep=0.5mm,pos=0.45]
(boucleap) {$a'$} (1a)
(2b) edge [in=110,out=70,loop,swap] node  [inner sep=0.5mm,pos=0.55]   (bouclej){$j$} (2a)
(3a) edge [in=-70, out=-110,loop,swap] node  [inner sep=0.5mm,pos=0.55]   (bouclek){$k$} (3a)

(123-a-1a) edge[matched edges, bend left=50] node {} (1a-b-2b)
(123-a'-1a) edge[matched edges, bend right=50] node {} (1a-b-3a)
(bouclea) edge[matched edges, bend left=50] node {} (1a-b-2a)
(boucleap) edge[matched edges, bend right=80] node {} (1a-b-3b)
(2b)  edge[bend right=40, inner sep=0mm] node {$i$} (1b)
(3a)  edge[bend left=40,swap,inner sep=0mm] node {$i$} (1b)
(1b)  edge[bend left=10] node [inner sep=0.5mm,pos=0.5] (1b-a-1a) {$a$} (1a)
(1b)  edge[bend right=10,swap] node [inner sep=0.5mm,pos=0.5] (1b-a'-1a) {$a'$} (1a)
(1b-a-1a) edge[matched edges, bend left=0] node {} (1a-b-2b)
(1b-a'-1a) edge[matched edges, bend right=0] node {} (1a-b-3a)
(2a) edge [in=60,out=100, loop] node {$j$}  (2a)
(3b) edge [in=-100,out=-60, loop] node {$k$}  (3b)
(1b)  edge[bend left=0, swap,inner sep=0mm] node {$b$} (2b)
(1b)  edge[bend right=0,inner sep=0mm] node {$b$} (3a);

\end{tikzpicture}
\caption{A Dyck automaton $\A$ (on the left) over $A=A_c \sqcup A_r \sqcup A_i$ with $A_c= \{a,a'\}$, $A_r= \{b\}$ and $A_i=\{i,j,k\}$. The left reduction of $\A$ (on the right) over the same
tri-partitioned alphabet. Matched edges are linked with a dotted line
and each state is represented by its summary set $S$ of pairs of edges
and the set $R$. 
      }\label{figure.exampleLetReduction}
\end{figure}
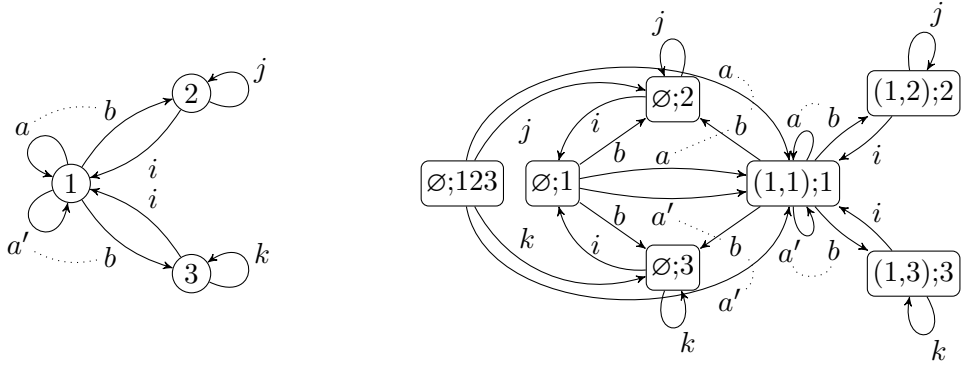
\end{example}

\begin{example} \label{example.example2}
The sofic-Dyck shift $X$ presented by the Dyck automaton $\A$ of Figure
\ref{figure.nondeterminizable} has no deterministic presentation. 
\begin{figure}[htbp]
    \centering
\begin{tikzpicture}[scale=0.4]
\node[state]            (3)         at   (0,0)                 {3};
\node[state]            (1) [right =2.5cm of 3]            {1};
\node[state]            (2) [right =2.5cm of 1]            {2};
\path     
(1)  edge[bend right] node {$i$} (2)
(2)  edge[bend right] node {$j$} (1)
(1)  edge[bend right] node {$i$} (3)
(3)  edge[bend right] node {$k$} (1)
(3) edge [in=120, out=60, loop] node  [inner sep=1mm,pos=0.5]   (x) {$b$} (3)
(1) edge [in = 120, out = 60, loop] node  [inner sep=1mm,pos=0.5]   (a) {$a$} (1)
(2) edge [in=120, out=60, loop] node  [inner sep=1mm,pos=0.5]   (b) {$b$} (2)
(a) edge[matched edges, bend left] node {} (b);
\end{tikzpicture}
\caption{A Dyck automaton $\A$ over $A$, with $A_c= \{a\}$, $A_r=
  \{b\}$ and $A_i= \{i,j,k\}$, presenting a sofic-Dyck shift which has
  no deterministic presentation. Matched edges are linked with a
  dotted line.
      }\label{figure.nondeterminizable}
\end{figure}
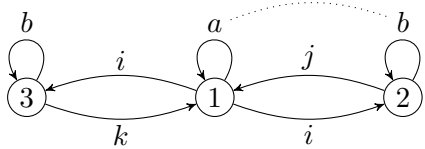
Let us briefly give a sketch of the proof of this fact.
Let $\B$ be a deterministic Dyck automaton over $\A$ accepting the same
shift $X$.  For any positive integers $n, m, r$, the words
$(a^{n+m}ib^nj)^r(ib^mj)^r(ib^mk)^r$ are blocks of $X$ and are thus
factors of labels of
admissible paths in $\B$, the edges labeled by $b$ of the path
labeling $(ib^mj)^r$ being matched with the edges labeled by $a$ of the path
labeling $(a^{n+m}ib^nj)$. 
If $\B$ is finite and deterministic, this implies that there are $m,r,s > 0$ and $n' > n >0$ such that
$(a^{n'+n+m}ib^nj)^r(ib^mj)^r(ib^mk)^r$ is a block of $X$, a contradiction.
\end{example}


\subsection{Visibly pushdown shifts}

In this section we show that the class of sofic-Dyck shifts is
the class of \emph{visibly pushdown shifts}, \ie the subshifts whose
set of blocks are factorial extensible visibly pushdown languages.

The class of visibly pushdown languages of finite words can be described either by
pushdown automata or by context-free grammars. 

A \emph{visibly pushdown automaton} on finite words over
$A=A_c \sqcup A_r \sqcup A_i$ is a tuple $M= (Q,I,\Gamma,\Delta,F)$ where Q is a finite set
of states, $I \subseteq Q$ is a set of initial states, $\Gamma$ is a
finite stack alphabet that contains a special bottom-of-stack symbol
$\bot$, $\Delta \subseteq (Q \times A_c \times Q \times (\Gamma
\setminus \{\bot\})) \cup (Q \times A_r \times \Gamma \times Q)
\cup (Q \times A_i \times Q)$, and $F \subseteq Q$ is a set of final
states.

A transition $(p,a,q,\gamma)$, where $a \in A_c$ and $\gamma \neq\bot$,
is a push-transition. On reading $a$, the stack symbol $\gamma$ is pushed onto the
stack and the control changes from state $p$ to $q$. 
A transition $(p,a,\gamma,q)$ is a pop-transition. The symbol $\gamma$
is read from the top of the stack and popped. If $\gamma = \bot$, the
symbol is read but not popped. A transition $(p,a,q)$ is a local action.

A stack is a nonempty finite sequence over $\Gamma$ starting with $\bot$.
A \emph{run} of $M$ labeled by $w=a_1 \ldots a_k$ is a sequence
$(p_0,\sigma_0) \cdots (p_k,\sigma_k)$ where $p_i \in Q$, $\sigma_0 = \bot$,
$\sigma_i \in   (\Gamma \setminus\{\bot\})$ for $1 \leq i \leq k$,
and such that:
\begin{itemize}
\item If $a_i$ is call symbol, then there are $\gamma_i \in \Gamma$ and
  $(p_{i-1},a_i,p_i,\gamma_i) \in \Delta$ with $\sigma_{i}= 
  \sigma_{i-1} \cdot \gamma_i$.
\item  If $a_i$ is a return symbol, then there are $\gamma_i \in \Gamma$ and
  $(p_{i-1},a_i,\gamma_i,p_i) \in \Delta$ with either $\gamma_i \neq
  \bot$ and $\sigma_{i} \cdot \gamma_i = \sigma_{i-1}$ or $\gamma_i = \bot$ and
  $\sigma_i=\sigma_{i-1}=\bot$. 
\item  If $a_i$ is an internal symbol, then $(p_{i-1},a_i,p_i) \in \Delta$ and $\sigma_i=\sigma_{i-1}$.
\end{itemize}
A run is \emph{accepting} if $p_0 \in I$, $\sigma_0=\bot$, and the last state is final, \ie $p_k \in
F$. A word over $A$ is \emph{accepted} if it is the label of an
accepting run. The language of words accepted by $M$ is denoted by $L(M)$.
The language accepted by a visibly pushdown automaton is called 
a \emph{visibly pushdown language}.

We will use also the
following grammar-based characterization (see \cite[Section~5]{AlurMadhusudan2005}).

A context-free grammar over an alphabet $A$ is a tuple $G=(V,S,P)$,
where $V$ is a finite set of variables, $S \in V$ is a start variable
and $P$ is a finite set of rules of the form $X \xrightarrow{}
\alpha$ such that $X \in V$ and $\alpha \in (V \cup A)^*$. The
semantics of the grammar $G$ is defined by the derivation relation
$\xrightarrow{}$ over $(V \cup A)^*$. If $X \xrightarrow{} \alpha$ is a
rule and $\beta,\beta'$ are words of $(V \cup A)^*$, then 
$\beta X \beta' \xrightarrow{} \beta \alpha\beta'$ holds.
The language \emph{accepted by the grammar} $G$, denoted $L(G)$ is the
set of words $u$ in $B^*$ such that $S \xRightarrow{*} u$, where
$\xRightarrow{*}$ is the transitive closure of the relation~$\xrightarrow{}$.

Let $A$ be a tri-partitioned alphabet.
A context-free grammar $G=(V,S,P)$ over $A$ is a \emph{visibly
  pushdown grammar} with respect to $A$
if the set $V$ of variables is partitioned into two disjoint sets
$V^0$ and $V^1$, such that all rules in $P$ are of one 
of the following forms
\begin{itemize}
\item $X \rightarrow \varepsilon$;
\item $X \rightarrow aY$, such that if $X \in V^0$, then $a \in A_i$ and
  $Y \in V^0$;
\item $X \rightarrow aYbZ$, such that $a \in A_c$, $b \in A_r$, $Y \in V^0$,
  and if $X \in V^0$, then $Z \in V^0$.
\end{itemize}
The variables in $V^0$ derive only Dyck words.
The variables in $V^1$ derive words that may contain unmatched call
letters as well as unmatched return letters. 
In the rule $X \rightarrow aY$, if $a$ is a call it is
unmatched and the variable $X$ must be in $V^1$ if $a$ is a call or
return.
In the rule $X \rightarrow
aYbZ$, the symbols $a$ and $b$ are the matching call and return.
The words generated by $Y$ belong to
$V^0$ and thus are Dyck words. Furthermore, if $X$ is required to
generate Dyck words, then $Z$ also.

It is shown in \cite{AlurMadhusudan2005} that a language is visibly
  pushdown language if and only if it is accepted by
  a visibly pushdown grammar.


\begin{proposition} \label{proposition.admissibleVPL}
The set of labels of finite admissible paths of a Dyck automaton is a
visibly pushdown language.
\end{proposition}
\begin{proof}
Let $\A=(\G,M)$ be a Dyck automaton over $A$, where $\G=(Q,E)$.

We define a visibly pushdown
automaton $V= (Q,I,\Gamma,\Delta,F)$ over $A$, where $I=F=Q$ and $\Gamma$
is the set of edges of $\A$. The set of transitions $\Delta$ is
obtained as follows.
\begin{itemize}
\item If $(p,a,q) \in E$ with $a \in A_c$, then $(p,a,q,(p,a,q)) \in
\Delta$. 
\item If $(p,b,q) \in E$ with $b \in A_r$, then $(p,b,\gamma,q) \in
\Delta$ for each call edge $\gamma$ which is matched with
the return edge $(p,b,q)$.
\item If $(p,\ell,q) \in E$ with $\ell \in A_i$, then $(p,\ell,q) \in \Delta$.
\end{itemize}
Let $w$ be a finite word over $A$. There is run $(p_0,\sigma_0) \cdots (p_k,\sigma_k)$  in $V$ 
labeled by $w$ such that $\sigma_0=\bot$, $p_0=p$ and $p_k=q$
if and only if $w$ be the label of an admissible path $\pi$ of $\A$
going from $p$ to $q$. Thus $w$ is the label of an admissible path of $\A$ if and only if
it is the label of an accepting run of $V$, which proves the proposition.

\end{proof}

In order to prove that the set of blocks of sofic-Dyck shift is a visibly
pushdown language, we have to prove that the subset of words
labeling a finite admissible path which are extensible to labels of
bi-infinite admissible paths is also a visibly pushdown language. 

Let $L$ be a language of finite words over $A$. We denote by 
$\E(L)$ the set words $w \in L$ such that, for any integer $n$,
there are words $u,v$ of length greater than $n$ such that $uwv \in
L$. Note that $\E(L)$ is a factorial language when $L$ is factorial.
This set is called in \cite{CulikYu1991} the \emph{bi-extensible}
subset of $L$. 

We show below that the bi-extensible subset of a
factorial visibly pushdown language is a visibly pushdown language.
It is shown in \cite{CulikYu1991} that it is not true that the
bi-extensible subset of  a context-free language is a context-free
language but the result holds for factorial context-free languages.
We prove a similar result for factorial visibly pushdown languages.

We first recall the following pumping lemma (see \cite[Lemma 5.6]{CulikYu1991}).
\begin{lemma} \label{lemma.pumping}
Let $G=(V,S,P)$ be a context-free grammar and $L = L(G)$. Then for any integer
$t > 0$, there exists an integer $p(t)$ such that for each $z \in L$
and any set $K$ of distinguished positions in $z$, if $|K| \geq p(t)$, then there is a decomposition
$z=ux_1 \cdots x_twy_1\cdots y_tv$
such that
\begin{itemize}
\item There exists a variable $X \in V$ such that
$$
S \xrightarrow{*} uXv \xrightarrow{*} ux_1Xy_1v \xrightarrow{*} \cdots \xrightarrow{*} 
 ux_1\cdots x_tXy_t \cdots y_1v  \xrightarrow{*} ux_1\cdots x_twy_t \cdots y_1v   .$$
\item For any $i_1, \ldots, i_t$, we have 
$ux_1^{i_1}\cdots x_t^{i_t}wy_t^{i_t} \cdots y_1^{i_1}v \in L$.
\item If $K(x)$ denotes the distinguished positions of $K$ in a word $x$,
then either $K(u), K(x_1), \ldots K(x_t), K(w) \neq \emptyset$, or 
$K(w), K(y_t), \ldots K(y_1)$, $K(v) \neq \emptyset$.
We also have $|K(x_1) \cup \ldots K(x_t)\cup K(w) \cup K(y_t) \cup
\ldots K(y_1)|\leq p(t)$.
\end{itemize}
\end{lemma}

\begin{proposition}\label{proposition.biextensible}
If $L$ is a factorial visibly pushdown language, then $\E(L)$ is a factorial visibly pushdown language.
\end{proposition}
\begin{proof}
Let $G=(V,S,P)$ be a visibly pushdown grammar over $A$ accepting $L$.
We define a grammar $G'=(V\cup\{X_i\},S,P')$ over $A'=(A_c\cup
\{\$_1\},A_r\cup \{\overline{\$_1}\},A_i\cup \{\$_0\})$
obtained by adding the following rules to
$G$:
\begin{itemize}
\item $X \rightarrow \$_1X\overline{\$_1}X_1$ and $\X_1 \rightarrow \varepsilon$, for each $X \in V^0$ such that $X \xrightarrow[G]{*} uXv$ with $u,v \in
A^+$,
\item  $X \rightarrow \$_0X$, for each $X \in V$ such that $X \xrightarrow[G]{*} uX$ with $u \in
A^+$.
\end{itemize}

Note that it is not possible to have a rule $X \in V$ such that $X
\xrightarrow{*} Xu$
with $u \in A^+$.
The grammar $G'$ is a visibly pushdown grammar over $A'$.

Let $L_1= \{ w \in A^* \mid  \exists w_i \in A^*, \: w_1\$ w_2ww_3\$'w_4 \in L(G'), \$ =\$_0 \text{ or } \overline{\$_1},
\$'= \$_0 \text{ or } \$_1\}$, $L_2 =
\{ w \in A ^* \mid \exists w_i \in A^*, w_1\$_1w_2ww_3\overline{\$_1}w_4, \in L(G')\}$  and $L_3= L_1 \cup L_2$.

Let us prove that $L_3 \subseteq \E(L)$.
We first consider a word $w \in L_2$ such that  $w_1\$_1w_2ww_3
\overline{\$_1}w_4 \in L(G')$. Then $w_2ww_3$ is generated in $G$ by some variable $X \in V^0$ such that $X
\xrightarrow{*} uXv$, for $u,v \in A^+$. Thus, for any integer $n$, we have $u^nw_2ww_3v^n \in
L$. Thus $w \in \E(L)$.

Let us consider a word $w \in L_1$ such that  $w_1\overline{\$_1} w_2ww_3
\$_1w_4 \in L(G')$. Thus there are words $u_1,u_2,u_3,u_4$ such that $u_1\$_1u_2
\overline{\$_1} w_2ww_3 \$_1 u_3 \overline{\$_1}u_4\in L(G')$. It follows that
there are words $x,y,z,t \in A^+$ such that $u_1x^nu_2y^nw_2ww_3$ $z^nu_3t^nv_4 \in
L$, for any positive integer $n$. Thus $w \in \E(L)$.

We now consider the case where  $w \in L_1$ with  $w_1\$_0 w_2ww_3
\$_0w_4 \in L(G')$.  Then there are variables $X$, $Y$ such that
$X \xrightarrow{*} uX$ for some $u \in A^+$,
$Y \xrightarrow{*} vY$ with $v \in A^+$, such that
$S \xrightarrow{*} \alpha X\beta Y\gamma$ and $X\beta \xrightarrow{*} w_2ww_3$.
It follows that, for any positive integer $n$, we have
$u^nw_2ww_3v^n \in L$.  Hence $w \in \E(L)$.
The remaining cases are proved similarly.

We now prove that $\E(L) \subseteq L_3$. Let $z \in \E(L)$ of length
$t$. We choose $z_1,z_2 \in A^+$ of length greater than $p(t)$, where $p(t)$
is defined in Lemma \ref{lemma.pumping}, such that $z' = z_1zz_2 \in
L$. For technical reasons that will appear below, we also choose
$|z_2| > 4|K|(|z_1| + |z|)$.

We consider a set of distinguished
positions in $z_1$.
By Lemma \ref{lemma.pumping}, 
there is a variable $X$ in $V$ such that
\begin{equation*}
S \xrightarrow{*} uXv \xrightarrow{*} ux_1Xy_1v \xrightarrow{*} \cdots
 ux_1\cdots x_tXy_t \cdots y_1v 
  \xrightarrow{*}ux_1\cdots x_twy_t \cdots y_1v = z'.
\end{equation*}
Let $T$ be the induced derivation tree and $T'$ be the subtree of $T$ (labeled
by $X$) generating $w$.
Let $\pi$ be the path going from the
root of $T$ to the parent of the root of $T'$. The length $\ell$ of
$\pi$ is at most $|ux_1\cdots
x_t|$ since all rules of $G$ produce either the empty word or a non empty
word over $V \cup A$ with a terminal symbol on the left.

At least one of two following cases holds.
\begin{itemize}
\item $K(u),K(x_1),\ldots,K(x_t), K(w) \neq \emptyset$.
  \begin{itemize}
  \item If $z$ is a factor of $wy_t \cdots y_1$ and $X \rightarrow
    \$_1X\overline{\$_1}X_1$ is a rule of $G'$, then 
    $u\$_1x_1\cdots x_t$ $wy_t \cdots y_1\overline{\$_1}v \in L(G')$ and
    thus $z \in L_3$.
\item If $z$ is a factor of $wy_t \cdots y_1$ and $X \rightarrow
    \$_1X\overline{\$_1}X_1$ is not a rule of $G'$. Then $y_t \cdots y_1
    = \varepsilon$ and $X \rightarrow \$_0X$ is a rule of $G'$.
 Thus the word $\hat{z}$ obtained after inserting $\$_0$ between $u$ and $x_1$ is
 still in $L(G')$. 

Furthermore, $z$ is a factor of $w$. We set $w=
 w_1zw_2$. 
If $|w_2| < |K|$, 
$|y_t \cdots y_1v| = |z_2| - |w_2| > 
4|K|\times|z_1z| - |K| \geq 3|K|\times|ux_1\cdots x_t| \geq 3\ell |K|$.
We denote by $R$ the set of nodes in $T$ which are children of nodes
of $\pi$ on the right of $\pi$ and thus generate $y_t \cdots y_1v$.
The size of the set $R$ is at most $3\ell$ since all rules of $G$ have
an arity at most $4$. 

At most $3 \ell$ variables generating a sequence of length greater
than $3 \ell |K|$, there is a variable $Y$ in $R$ such that $Y$ generates
a factor of length at least $|K|$ of $y_t \cdots y_1v$ which is a
factor of $z_2$. We do a second pumping for words generated by $Y$
using distinguished positions on $y_t \cdots y_1v$ and get 
that the word obtained from $\hat{z}$ after inserting either $\$_0$ or $\$_1$
in $y_t \cdots y_1v$ is still in $L(G')$.

If $|w_2| \geq |K|$,  we do a second pumping for words generated by $X$
using distinguished positions on $w_2$ and get 
that the word obtained from $\hat{z}$ after inserting either $\$_0$ or $\$_1$
in $w_2$ is still in $L(G')$.

\item If $z$ is a factor of $y_t \cdots y_1v$, then $w$ can be
  replaced either by $\overline{\$_1}w$ or by 
  $\$_0w$ and gives a word $\hat{z}$.  
A similar argument as above for a second pumping still holds.
We have this time $|z_2| > 
4|K|\times|z_1z| \geq 4|K|\times|ux_1\cdots x_t| \geq 3\ell |K|$.
Hence there is a variable $Y$
  in $R$ such that $Y$ generates a factor of length at least $|K|$ of
  $z_2$.  
We do a second pumping for words generated by $Y$
using distinguished positions on $z_2$ and get 
that the word obtained from $\hat{z}$ after inserting either $\$_0$ or $\$_1$
in $z_2$ is still in $L(G')$.

\item Otherwise $z$ is a factor of $wy_t \cdots y_1v$ and $z$ crosses $w$,
  $y_t \cdots y_1$ and $v$. Then $|y_t \cdots y_1 | < |z| =t$. 
 Thus there is $1 \leq i \leq t$ such that $y_i = \varepsilon$.
So we can replace $x_i$ by $\$_0x_i$ obtaining $\hat{z}$.
Again here there is a variable $Y$
  in $R$ such that $Y$ generates a factor of length at least $|K|$ of
  $z_2$. We do a second pumping for words generated by $Y$
using distinguished positions on $z_2$ and get 
that the word obtained from $\hat{z}$ after inserting either $\$_0$ or $\$_1$
in $z_2$ is still in $L(G')$.

  \end{itemize}
\item $K(w),K(y_t),\ldots,K(y_1), K(v) \neq \emptyset$. Then $z$ is a
  factor of $v$ since the distinguished positions are on $z_1$.  Then $w$ can be
  replaced either by $\overline{\$_1}w$ or by 
  $\$_0w$.  The second pumping is done as the in the last item of the
  previous case.
\end{itemize}
We obtain that for any $z \in \E(L)$, there is in $L(G')$ either a word of the form
$w_1\$ w_2zw_3\$'w_4$ with  $\$ =\$_0 \text{ or } \overline{\$_1}$,
$\$'= \$_0 \text{ or } \$_1$ or a word of the form
$w_1\$_1w_2zw_3\overline{\$_1}w_4$. Thus $z \in L_3$.
Hence $L_3 = \E(L)$.
 
We now show that $L_3$ is a visibly pushdown language. 
Indeed, let us show that
$L_{\$\$'}= \{ w \in A^* \mid  w_1\$ w_2ww_3\$' w_4 \in L(G')\}$
is visibly pushdown. 

Let $L' = \Fact(L(G')) \cap
\$A^*\$'$, where $\Fact(L(G'))$ denotes the set of factors of
$L(G')$, and  let $L" = \Fact(L') \cap A^*$.
Since the class of visibly pushdown languages is closed by prefix
and suffix, it is closed by factor. Hence the languages $L'$ and $L"$ are visibly
pushdown. We have $L_{\$\$'} = L"$.

As a consequence, $L_{\$\$'}$ is visibly pushdown.
The class of visibly pushdown languages being closed by union,
we get that $L_3$ is visibly pushdown.

Note that that $G'$ can be constructed in an effective way since it is
decidable whether $X \xrightarrow{*} uXv$ or $X \xrightarrow{*} uX$
for some words $u,v \in A^+$.
\end{proof}

\begin{theorem} \label{theorem.VPL}
Let $X$ be a sofic-Dyck shift. Then $\B(X)$ is a visibly pushdown
language. Conversely, if $L$ is a factorial extensible visibly pushdown
language, then $\B^{-1}(L)$ is a sofic-Dyck shift.
\end{theorem}

\begin{proof}
Let $X$ be the sofic-Dyck shift presented by a Dyck automaton $\A$.
By Propositions \ref{proposition.admissibleVPL}, the set $L$ of labels
of finite admissible paths of $\A$ is a visibly pushdown
language. 
By \ref{proposition.biextensible}, the language $\E(L)$ also. 
By Lemma \ref{lemma.compacity}, we have $\B(X)= \E(L)$ and thus
$\B(X)$ is a visibly pushdown
language.

Conversely, let $L$ be a factorial extensible visibly pushdown 
language. Let $G=(V,S,P)$ be a visibly pushdown grammar over $A$
accepting $L$. We may assume that variables that do not generate any
word are discarded.
We define a Dyck automaton
$\A=(\G,M)$ with $\G=(V \cup (V\times (\{\$\} \cup (A \times V)),E)$ as follows.  
We denote below by $(X,\circ)$ any state which is either $X$ or $(X,\$)$, or $(X, (a,Y))$.
\begin{itemize}
 \item If $X \rightarrow \ell Y \in P$ with $\ell \in A_i$, then
$((X,\circ),\ell,Y) \in E$.
\item If $X \rightarrow a Y \in P$ with $a \in A_c$, then 
$((X,\circ),a,(Y,\$)) \in E$.
\item If $X \rightarrow a Y b Z \in P$, then 
$((X,\circ),a,(Y,(b,Z))) \in E$.
\item If $X \rightarrow b Y \in P$ with $b \in A_r$, then
$((X,\circ),b,Y) \in E$.
\item If $b \in A_r$, $Z \xrightarrow{} \varepsilon$ and $Z \in V^0$, then 
$((Z,\circ),b,T) \in E$ for any $T \in V$.
Each of these edges is also matched with each edge of the form $((X,\circ)),a,(Y,(b,T)))$.
\end{itemize}
Note that all states $(X,\circ)$ have the same outgoing edges. 
A state $(X,\circ)$ is \emph{nullable} if $X$ generates the empty word.

We claim that if $w$ is a word generated by $X$ in $G$, there is 
an admissible path in $\A$ labeled by $w$ from $X$ to some nullable state $T$ or
$(T,\$)$. 

The proof is by induction on the size of $w$. Let us first consider
the case $w=\varepsilon$. If $w$ is generated by $Z$, then 
$Z \rightarrow \varepsilon$ is a rule of $G$. Thus the claim is true.

If $w$ is nonempty, since $w$ is generated by $X$, then either $X \rightarrow aY \in P$,
$w=aw_1$  with $w_1$ is generated by $Y$ and $a$ is not matched with
symbols of $w_1$, or $X \rightarrow aY bZ\in P$, $w=aw_1bw_2$ and $w_1,w_2$ are generated by $Y$ and $Z$ respectively, 
with $Y \in V^0$.

In the first case, there is an edge $(X,a,(Y,\$))$. By induction, there is 
an admissible path in $\A$ from $Y$ to some nullable state $T$ or
$(T,\$)$. Thus there is an admissible path labeled by $w_1$ 
from $(Y,\$)$ to some nullable state $T$ or $(T,\$)$ and thus there
is an admissible path
labeled by $w$ from $X$ to some nullable state $T$ or $(T,\$)$.

In the second case, there is an edge $(X,a,(Y,(b,Z)))$.
By induction, there is an admissible path labeled by $w_1$ 
from $(Y,(b,Z))$ to some nullable state $T$ or $(T,\$)$.
There is also an admissible path labeled by $w_2$
from $(Z,\circ)$ to some nullable state $U$ or $(U,\$)$
and an edge $((T,\circ),b,Z)$. Thus we obtain the path
\begin{equation*}
X \xrightarrow{a} (Y,(b,Z)) \xrightarrow{w_1} T (\text{or } (T,\$))
\xrightarrow{b}  Z \xrightarrow{w_2} U (\text{or } (U,\$)).
\end{equation*}
Since $T$ is nullable and in $V^0$, any edge $((T,\circ),b,Z)$ is
matched with $(X,a$, $(Y,(b,Z)))$,
this path is an admissible path labeled by $w$ going 
from $X$ to either $U$ or $(U,\$)$.
Thus $L$ is included in the set of labels of admissible paths of $\A$.

Conversely, let $w$ be the label of an admissible path $\pi$ in
$\A$ starting at a state $(X,\circ)$. 
Then $w$ is a prefix of a word generated by $X$ in $G$. 
If $w$ is moreover a Dyck word and $X$ is nullable, then $w$ is generated by $X$. 


The proof is again by induction on the size of $w$. Note that it holds
for the empty word.
We first decompose $\pi$ into one of the following paths:
\begin{enumerate}
\item $(X,\circ) \xrightarrow{a} Y  \xrightarrow{w_1} (U,\circ), \text{
  with } a \in A_i,$
\item  $(X,\circ) \xrightarrow{a} (Y,\$)  \xrightarrow{w_1} (U,\circ)$, 
\text{$a \in A_c$ not matched with letters of $w_1$}
\item 
$(X,\circ) \xrightarrow{a} Y \xrightarrow{w_1} (U,\circ), \text{
  with }  a \in A_r,$
\item
$(X,\circ) \xrightarrow{a} (Y,(b,Z)) \xrightarrow{w_1} (U,\circ)$,
\text{$a$ not matched with letters of $w_1$},
\item
$(X,\circ) \xrightarrow{a} (Y,(b,Z)) \xrightarrow{w_1} (T,\circ)
\xrightarrow{b} Z \xrightarrow{w_2} (U,\circ), 
\text{ with } a \in A_c,  b \in A_r, w_1 \text{ is a Dyck word}$.
\end{enumerate}
In Cases (1) to (3), by induction, $w_1$ is a prefix of a word generated by
$Y$ and there is a rule $X \rightarrow a Y$ in $G$. Thus $aw_1$
is a prefix of a word generated by~$X$. 
If $w$ is a Dyck word and $U$ is nullable, then $w = aw_1$, where
$w_1$ is a Dyck word $a \in A_i$.  By induction hypothesis, the word $w_1$
is generated by $Y$ and thus $w$ is generated by $X$.



In Case (4), by induction, $w_1$ is a prefix of a word generated by
$Y$ and there is a rule $X \rightarrow aYbZ$ in $G$. Thus $aw_1$
is a prefix of a word generated by $X$. The word $w$ is never a Dyck word.

In Case (5), there is a rule $X \rightarrow aYbZ$ in $G$ and the
edges $((X,\circ),a,$ $(Y,(b,Z))$ and $((T,\circ),b,Z)$ are matched
Thus $T$ is nullable and in $V^0$.
By induction $w_2$ is a prefix of a word generated by $Z$. Since $w_1$ is
a Dyck word, by induction again, $w_1$ is generated
by $Y$. It follows that $w$ is a prefix of a word generated by $X$.
If $w$ is a Dyck word, then $w_2$ is a Dyck word and thus $w_2$
is generated by $Z$. As a consequence, $w$ is generated by $X$.

Thus labels of admissible paths of $\A$ are prefixes of words of $L$.
Since $L$ is factorial, they belong to $L$. As a consequence $L$ is the set
of labels of finite admissible paths of $\A$. By definition,
$\B^{-1}(L)$ is the set of infinite sequences whose finite factors
belong to $L$ and thus $\A$ presents $\B^{-1}(L)$. 
\end{proof}


This gives the following characterization of sofic-Dyck shifts.
\begin{theorem} \label{theorem.VPL2}
Sofic-Dyck shifts over $A$ are shifts $\X_F$ where $F$ is a visibly
pushdown language over $A$.
\end{theorem}
\begin{proof}
If $X$ is a sofic-Dyck shift over $A$, then Theorem \ref{theorem.VPL} says
that $\B(X)$ is a visibly pushdown language over $A$. Let $F = A^*
\setminus \B(X)$. Since visibly pushdown languages are closed y
complementation, $F$ is visibly pushdown and $X = X_\F$. 

Conversely, if $X = \X_F$ where $F$ is a visibly
pushdown language. Let $L = A^* \setminus F$ which is a factorial
visibly pushdown language. 
The set $\B(X)$ is the set of extensible words of $L$ and is thus
visibly pushdown. Thus $X$ is sofic-Dyck.
\end{proof}

\begin{proposition}
It is decidable whether a sofic-Dyck shift is empty.
\end{proposition}
\begin{proof}
Let $X$ be a sofic-Dyck shift. By Theorem \ref{theorem.VPL}, 
the set of blocks of $X$ is generated by a context-free grammar
which is furthermore computable from some Dyck automaton accepting
the sofic-Dyck shift. Since the emptiness is decidable for a
language generated by a context-free grammar, the emptiness of $X$ is decidable. Indeed, $X$ is nonempty
if and only if its set of blocks is nonempty.
\end{proof}

\section{Zeta function of sofic-Dyck
  shifts} \label{section.zeta}

Zeta functions count the periodic orbits of subshifts and constitute
stronger invariants by conjugacies than the entropy (see
\cite{LindMarcus1995}).

In this section, we give an expression of the zeta function of a 
sofic-Dyck shift which extends the formula obtained by Krieger and
Matsumoto in \cite{KriegerMatsumoto2011} for Markov-Dyck shifts.
The proof of Krieger and Matsumoto is based on Markov-Dyck
codes which encode periodic sequences. We use a similar
encoding to compute the zeta function of sofic-Dyck shifts.

As counting periodic points for sofic shifts is trickier than for
shifts of finite type, counting periodic points of sofic-Dyck shifts
is also trickier than for Markov-Dyck of finite-type-Dyck shifts.

\subsection{Definition and general formula} \label{section.zetaDefinition}

The \emph{zeta
  function}  $\zeta_X(z)$ of the shift $X$ is defined as the
zeta function of its set of periodic patterns, \ie 
\begin{equation*}
\zeta_X(z) = \exp \sum_{n \geq 1} p_n \frac{z^n}{n},
\end{equation*}
where $p_n$ the number of sequences of $X$ of period $n$, \ie of sequences 
$x$ such that $\sigma^n(x) = x$. Note that $n$ may not be the smallest
period of $x$.

Call \emph{periodic pattern} of $X$ a word $u$ such that the
bi-infinite concatenation of $u$ belongs to $X$ and denote $P(X)$
the set of periodic patterns of $X$. 
These definitions are extended to $\sigma$-invariant 
sets of bi-infinite sequences which may not be
shifts (\ie which may not be closed subsets of sequences).



Let $\A$ be a Dyck automaton over a tri-partitioned alphabet $A$. 

We say that a Dyck word $w$ over $A$ is \emph{prime} if it is nonempty
and any Dyck word prefix of $w$ is $w$ or the empty word. We denote by
$\Prime(A)$ the set of prime Dyck words over $A$ and by $\Prime(X)$
the set of prime Dyck words which are blocks of a shift $X$.

We define the following matrices where $Q$ is the set of states of $\A$.
\begin{itemize}
\item  $C=(C_{pq})_{p,q \in Q}$ where $C_{pq}$ is the
set of prime Dyck words labeling an admissible path from
$p$ to $q$ in $\A$.

\item $M_c=(M_{c,pq})$, (\resp $M_r$) where $M_{c,pq}$ is the sum of call (\resp
return) letters labeling
an edge from $p$ to $q$ in $\A$.

\end{itemize}
Let $H$ be one of the matrices $C$, $C{M_c}^*$, $M_c$, ${M_r}^*C$ or $M_r$.
We call \emph{$H$-path} a path $(p_i,c_i,p_{i+1})_{i \in
  I}$ in $\A$, where $I$ is $\Z$ or an interval and $c_i \in
H_{p_ip_{i+1}}$. Note that an $H$-path is admissible. 
We denote by $\X_H$ be the $\sigma$-invariant set containing all
of sequences labeling a bi-infinite $H$-path of $\A$.

\begin{proposition} \label{proposition.P(X)}
Let $X$ be a the sofic-Dyck shift accepted by a Dyck automaton $\A$. We have
$P(X) = P(\X_{M_c}) \sqcup P(\X_{M_r}) \sqcup
((P(\X_{C{M_c}^*}) \cup p(\X_{{M_r}^*C})))$, and
$P(X_C) = P(\X_{C{M_c}^*}) \cap P(\X_{{M_r}^*C})$,
where $\sqcup$ denotes a disjoint union.
\end{proposition}
For a finite word $u$, we denote the \emph{balance} of $u$ by
$\bal(u)$. It is the difference between the number of letters of $u$ in $A_c$ and the number
of letters of $u$ in $A_r$. A word $u$ is \emph{positive} if $\bal(u)
>0$ and $\bal(v) \geq 0$ for any prefix $v$ of $u$. 
We
say that $u$ and $v$ are \emph{conjugate} if they are words $w,t$ such that
$u=wt$ and $v= tw$. 
\begin{proof}
Let us assume that a sequence $x$ of $X$ is equal to $u^\infty= \dotsm uu\cdot uu
\dotsm $. Let $u = u_0u_1 \cdots u_{n-1}$ where $u_i$ are letters. 
We consider the following three cases.
\begin{itemize}
\item If $\bal(u)= 0$, then $u$ is conjugate to a word in $\Prime(X)^*$ and thus $x$ is a periodic
point of $\X_C$. 
\item If $\bal(u) > 0$, then $u$ is conjugate to a word $v$ such that
$\bal(v_0 \ldots v_{i}) \geq 0$ for any $0 \leq i \leq n-1$.
If $v \in A_c^+$, then $x$ is a periodic point of $\X_{M_c}$.
If $v \notin A_c^+$, there are two indices $0 \leq m_1 < m_2 \leq n-1$
such that $\bal(v_0 \ldots v_{m_1}) = \bal(v_0 \ldots v_{m_2})$. 
Let $(m_1,m_2)$ two such indices with moreover $\bal(v_0 \ldots
v_{m_1}) = \bal(v_0 \ldots v_{m_2})$ minimal. 
Let $w = v_{m_1} \cdots v_{n-1}v_0 \cdots v_{m_1 -1}$. The word $w$ is
again a conjugate of $v$ and $u$.

Let $j_1$ be the largest integer less than or equal to $n-1$ 
  such that $w_{0} \cdots w_{j_1}$ has a suffix in $\Prime(X)$ and $i_1$ be
  the smallest integer such that $w_{i_1} \cdots w_{j_1} \in \Prime(X)^*$.
Then $w_{i_1} \cdots w_{n-1} \in \Prime(X)^+A_c^*$ and $w_0\cdots w_{i_1-1}$ is
a positive word. We define indices $i_2,j_2$ similarly for the word
$w_{0} \cdots w_{i_1-1}$ and thus iteratively decompose $w$ into a product of
words in $\Prime(X)A_c^*$. It follows that $x$ belongs to $\X_{CM_c^*}$.
\item If $\bal(u) < 0$, we denote by $\tilde{u}$ the word $u_{n-1}
  \cdots u_0$. By exchanging the roles played by call and return
  symbols, we have $\bal(\tilde{u}) > 0$ and thus 
either $\tilde{u}$ is conjugate to a word in $A_r^+$ or 
$\tilde{u}$ is conjugate to a word in $(\tilde{P}A_r^*)^+$, where
$\tilde{P}=\{\tilde{c} \mid c \in \Prime(X) \}$. We thus get that $u$
is conjugate to a word in  $(A_r^*\Prime(X))^+$ and $x$ belongs to $\X_{M_r^*C}$.
\end{itemize}
\end{proof}

As a consequence, we obtain the following expression of the zeta
function of a sofic-Dyck shift.

\begin{proposition} \label{proposition.decomposition}
Let $X$ be a sofic-Dyck shift presented by a Dyck automaton $\A$ and
$C$, $M_r$, $M_c$ defined as above from $\A$.
The zeta function of $X$ is
\begin{equation} 
 \zeta_{X}(z)  = \frac{\zeta_{\X_{C{M_c}^*}}(z) \zeta_{\X_{{M_r}^*C}}(z) 
\zeta_{\X_{M_c}}(z)  \zeta_{\X_{M_r}}(z)}{\zeta_{\X_C}(z)}
\end{equation}
\end{proposition}
\begin{proof}
The formula is a direct consequence of Proposition
\ref{proposition.P(X)} and of the definition of the zeta function.
\end{proof}

We recall below the notion of circular codes (see for instance
\cite{BerstelPerrinReutenauer2010}).
We say that a subset $S$ of nonempty words over $A^*$ is
\emph{a circular code} if for all $n,m \geq 1$ and
$x_1, x_2, \dotsc, x_n \in S$, $y_1, y_2,\dotsc, y_m \in
S$ and $p \in A^*$ and  $s \in A^+$, the equalities
\begin{align}
sx_2x_3\dotsm x_np &= y_1y_2\dotsm y_m,\label{equation.1}\\
x_1&=ps \label{equation.2}
\end{align}
implies 
\begin{align*}
n=m && p=\varepsilon && \text{ and } && x_i=y_i&& (1\leq i \leq n).
\end{align*}

\begin{proposition}
Let $A$ be a tri-partitioned alphabet. 
The sets $\Prime(A)$ and $\Prime(A){A_c}^*$ are circular codes.
\end{proposition}
\begin{proof}
We prove that $\Prime(A) {A_c}^*$ is circular. This implies that
its subset $\Prime(A)$ is circular.

Let us suppose that Equations \ref{equation.1} and \ref{equation.2}
imply $n=m$ and $x_i=y_i$ for $n+m < N$. Assume that Equations
\ref{equation.1} and \ref{equation.2} hold for some $n,m$ with
$n+m=N$.




Let us assume that $s \neq x_1$. Since $s$ is a prefix of some $y_1y_2
\cdots y_j$ and a suffix of $x_1$, we have $s \in \Prime(A)$ or $s \in
{A_c}^+$.  As $ps \in \Prime(A){A_c}^*$, we get that $p \in
\Prime(A)$ and $s \in {A_c}^+$. It implies that the balance of 
each nonempty prefix of $sx_2x_3\dotsm x_np$ is positive, in
contradiction with $y_1$ prefix of $sx_2x_3\dotsm x_np$.
Hence $s=x_1$ and $p = \varepsilon$. If $y_1\neq x_1$, one of these
two words is a prefix of the other. Let us assume that $x_1 = y_1z$
with $z \in A_c^*$. Then $zx_2x_3\dotsm x_n$ is positive, a
contradiction with that fact that it has $y_2$ as prefix.   Thus
$x_1=y_1$. By iteration of this process, we get $n=m$ and $x_i=y_i$.
\end{proof}

The notion of circular matrix below extends the classical notion of circular codes.
We say that the matrix $(H_{pq})_{p,q \in Q}$, where each $H_{pq}$ is a
set of nonempty words over $A$ is \emph{circular} if for all $n,m \geq 1$ and
$x_1 \in H_{p_0,p_1}, x_2 \in H_{p_1,p_2}$, $\dotsc, x_n \in
H_{p_{n-1}p_{0}}$, $y_1 \in H_{q_0,q_1}, y_2 \in H_{q_1,q_2},\dotsc, y_m \in
H_{q_{m-1}q_{0}}$ and $p \in A^*$ and  $s \in A^+$, the equalities
\begin{align}
sx_2x_3\dotsm x_np &= y_1y_2\dotsm y_m,\label{equation.4}\\
x_1&=ps \label{equation.5}
\end{align}
implies 
\begin{align*}
n=m && p=\varepsilon && \text{ and } && x_i=y_i && (1\leq i \leq n).
\end{align*}

\begin{proposition}
Let $\A$ be a Dyck automaton. The matrices $C$, $M_c$ and $C{M_c}^*$
defined from $\A$ are circular matrices. 
\end{proposition}
\begin{proof}
It is a direct consequence of the fact that $\Prime(A)$, $A$ and
$\Prime(A)A_c^*$ are circular codes.
\end{proof}

We say that $\A$ is \emph{left reduced} (\resp \emph{right reduced}) 
it is the left (\resp right) reduction of some Dyck automaton.

We say that $\A$ is 
\emph{$H$-deterministic} if and only if for any two (admissible)
$H$-paths sharing the same start and label are
equal.

\begin{proposition}
If $\A$ is left reduced, it is $H$-deterministic when $H$ is $M_c$, $C$ or $C{M_c}^*$.
\end{proposition}
\begin{proof}
The Dyck automaton $\A$ is $M_c$-deterministic by construction. It is $C$-deterministic
by Proposition \ref{proposition.DyckDeterministic}.
\end{proof}

One proves similarly that
\begin{proposition}
If $\A$ is right reduced, it is $H$-codeterministic for $H$ is $M_r$, $C$ or ${M_r}^*C$.
\end{proposition}

In order to count periodic sequences of sofic-Dyck shifts, we need some
machinery similar to the one used to count the periodic sequences of sofic
shifts (see for instance \cite{LindMarcus1995}).

Let $\A$ be a Dyck automaton over $A$ where $\A=(\G,M)$ with $\G=(Q,E)$.
Let $\ell$ be a positive integer. We fix an ordering on the states $Q$.
We define the Dyck automaton $\A_{\otimes\ell}= (\G_{\otimes\ell},M_{\otimes\ell})$ over
a new alphabet $A'$ where $\G_{\otimes\ell} =(Q_{\otimes\ell},E_{\otimes\ell})$ as follows.

\begin{itemize}
\item We set $A'=(A'_c,A'_r,A'_i)$ with $A'_c=A_c \cup \{-a \mid a
\in A_c\}$, $A'_r=A_r \cup \{-a \mid a
\in A_r\}$, and $A'_i=A_i \cup \{-a \mid a \in A_i\}$.
\item
We denote by $Q_{\otimes\ell}$ the set of ordered $\ell$-uples of
distinct states
of $Q$.
\item Let $P=(p_1,\dotsc,p_\ell)$, $R=(r_1,\dotsc,r_\ell)$, 
be two elements of $Q_{\otimes\ell}$. Thus $p_1 <\dotsm  < p_\ell$ and
$r_1 <\dotsm  < r_\ell$.
There is an edge labeled by $a$ from $P$ to $R$ in
$\A_{\otimes \ell}$ if and only if there are edges labeled by $a$ 
from $p_i$ to $p'_i$ for $1 \leq i \leq \ell$ and $R$ 
is an even permutation of $(p'_1,\dotsc,p'_\ell)$. If the permutation is odd we assign
the label $-a$. Otherwise, there is no edge with label $a$ or $-a$
from $P$ to $R$. 
\item 
We define $M_{\otimes\ell}$ as the set of pairs of edges 
$((p_1,\dotsc,p_\ell),a,(p'_1,\dotsc,p'_\ell)),$ 
and $((r_1,\dotsc,r_\ell)$, $\pm b, (r'_1,\dotsc,r'_\ell))$ of
$\A_{\otimes\ell}$ such that each edge
$(p_i,a,p'_i)$ is matched with $(r_i,b,r'_i)$ for $1 \leq i \leq \ell$.
\end{itemize}

We say that a path of $\A_{\otimes\ell}$ is \emph{admissible} if it is
admissible when the signs of the labels are omitted, the sign of the
label of a path being the product of the signs of the labels of the edges of
the path.

We denote by $C_{\otimes\ell,PP'}$ the set of signed prime Dyck words $c$ labeling
an admissible path in $\A_{\otimes\ell}$ from $P$ to $P'$. 
We denote by $C_{\otimes \ell}$ the matrix $(C_{\otimes \ell,PP'})_{P,P' \in
  Q_{\otimes\ell}}$ whose coefficients are sums of signed words of
$A^+$. More generally, if $H$ denotes one of the matrices
$C, C{M_c}^*,M_c,{M_r}^*C,M_r$ defined from $\A$, we 
denote by $H_{\otimes
  \ell}$ the matrix defined from
$\A_{\otimes\ell}$ similarly.

\subsection{Computation of the zeta function of $\X_H$}

Denote $\Z \llangle A \rrangle$ the set of
noncommutative formal power series over 
the alphabet $A$ with coefficients in $\Z$.
Let $\Z [\![ A ]\!]$ be the usual commutative algebra of formal power series
in the variables $a$ in $A$ and 
$\pi \colon \Z \llangle A \rrangle\rightarrow \Z [\![ A ]\!]$ be the
natural homomorphism. Let $S$ 
be a commutative or noncommutative series. One can write $S = \sum_{n \geq 0} [S]_n$
where each $[S]_n$ is the homogeneous part of $S$ of degree $n$. 
We denote by $\theta \colon \Z [\![ A ]\!]
\rightarrow \Z [\![ z ]\!]$ the homomorphism such that $\theta(a)=z$
for any letter $a \in A$.
The homomorphism $\theta$ and $\pi$ extends to matrices with
coefficients in $\Z \llangle A \rrangle$ and $\Z [\![ A ]\!]$ respectively.

\begin{proposition}  \label{proposition.trace}
Let $\A$ be a left reduced Dyck automaton and $H$ be one of the matrices
$C,C{M_c}^*,M_c$ defined from $\A$.
We have
\begin{align*}
\pi P_n(\X_H) 
&= \sum_{\ell=1}^{|Q|} (-1)^{\ell+1}\trace\sum_{1 \leq j \leq n} j
[\pi H_{\otimes\ell} ]_j [(1-\pi H_{\otimes\ell})^{-1}]_{n-j}.
\end{align*}
where $P_n(\X_H)$ is the set of periodic pattern of $\X_H$ of length~$n$. 
\end{proposition}
\begin{proof}
With a slight abuse of notations, we will say that a word $u$ belongs to
$H$ if $u$ is belongs to some $H_{pq}$.

Let $x$ be a periodic sequence of $\X_H$ of period $n \geq 1$.
We have $\sigma^n(x)=x$ if and only
if $x$ is a two-sided infinite concatenation of a word $w =
vx_2 \dotsm x_{k}u$ of
length $n$ with $x_i \in H$, $x_1=uv \in H$, and $v \neq
\varepsilon$. 
Let $j = |x_1|$.
The sequences $x,\sigma(x), \dotsc,
\sigma^{j-1}(x)$ are all distinct by circularity of the matrix $H$.
Since $\pi(w)= uvx_2\ldots x_k= x_1 \ldots x_k$, we get that
\begin{align*}
\pi P_n(\X_H) =  \sum_{1 \leq  j \leq n}\sum_{k \geq 1 }j E_{n,j,k}
= \sum_{1 \leq  j \leq n} j E_{n,j},
\end{align*}
where $E_{n,j,k}$ is the $H$-path labels of length $n$ which are concatenation of $k$
words $x _1 \ldots x_k$ of $H$ with $|x_1| = j$,
and $E_{n,j}$ is the union of the $E_{n,j,k}$.
The sets $E_{n,j,k}$ and $E_{n,j',k'}$ are disjoint for $k \neq k'$ or
for $j \neq j'$.

Let us denote by $D_{n,j,k}$ and $D_{n,j}$ the matrices
\begin{align*}
D_{n,j,k} &= [H]_j[H^{k-1}]_{n-j},\\
D_{n,j} &= \sum_{k \geq 1} D_{n,j,k}.
\end{align*}
Then $E_{n,j,k}$ (\resp $E_{n,j}$) is the set of labels
of $D_{n,j,k}$-paths (\resp $D_{n,k}$-paths).

Let $j$ be a fixed integer between
$1$ and $n$.  
Let us show that  
\begin{equation*}
\sum_{w \in E_{n,j}} w
= \sum_{\ell=1}^{|Q|} (-1)^{\ell+1} \trace((D_{n,j})_{\otimes\ell}).
\end{equation*}

Note that $w$ appears in $\trace((D_{n,j})_{\otimes\ell})$ for some integer $\ell$ with $1 \leq \ell
\leq  |Q|$ only if $w \in  E_{n,j}$.

Thus we can write
\begin{align}
\sum_{\ell=1}^{|Q|} (-1)^{\ell+1} \trace((D_{n,j})_{\otimes\ell})
&= \sum_{w \in E_{n,j}} c(w) \: w, \label{equation.cu}
\end{align}
where $c(w) \in \Z$. 

We will show that $c(w)=1$ for every word $w$ such that
$w \in E_{n,j}$.

Let $w$ be such a word. Since $\A$ is $H$-deterministic, the coefficient of each word in $H^k_{pq}$ for $k \geq 1$
and fixed states $p,q$, is at most one since there is at most one $H$-path
in $\A$ going from $p$ to $q$ and labeling a given word.
Hence the coefficient of $w$ in each $(D_{n,j})_{pq}$ is at most one
for fixed states $p,q$.

If $w \in E_{n,j}$
there must be at least one nonempty subset $R$ of $Q$ (of cardinal $m$)
on which $w$ acts as a permutation $\mu_w$ of $R$ induced by $D_{n,j}$, \ie such that the
coefficient of $w$ in $(D_{n,j})_{p\mu_w(p)}$ is one for each
$p \in R$. If two subsets have this property,
then does the union. Hence there is a largest subset $P \subseteq Q$
on which $w$ acts as a permutation. At this point we need a
combinatorial lemma used in \cite[Lemma 6.4.9]{LindMarcus1995}.
We recall its proof for the sake of completeness.

In the following lemma, the notation $\mu|_R$ means the restriction of
a permutation $\mu$ to a set of states $R$ and $\varepsilon(\mu)$ is the
signature of the permutation~$\mu$.
\begin{lemma}\cite[Lemma 6.4.9]{LindMarcus1995} \label{lemma.combi}
Let $\mu$ be a permutation of a finite set $P$ and let $\P=\{R
\subseteq P \mid R \neq \emptyset, \: \mu(R)=R\}$.
Then 
\begin{equation*}
\sum_{R \in \P} (-1)^{|R|+1} \varepsilon(\mu|_R)=1.
\end{equation*}
\end{lemma}
\begin{proof}[Proof of Lemma~\ref{lemma.combi}]
Recall that $P$ decomposes under $\mu$ into disjoint cycles, say $P_1,
\dotsm, P_d$. Thus each $\mu|_{P_i}$ is a cyclic permutation and so
\begin{equation*}
\varepsilon(\mu|_{P_i})=(-1)^{1+|P_i|}.
\end{equation*}
The nonempty sets $R \subseteq P$ for which $\mu(R)=R$ are exactly the
nonempty unions of sub-collections of $\{P_1, \dotsm, P_d\}$. Thus
\begin{align*}
\sum_{R \in \P} (-1)^{|R|+1} \varepsilon(\mu|_R)
&= \sum_{\emptyset \neq K \subseteq \{1,\dotsc,d\}}
(-1)^{1+|\cup_{k \in K} P_k|} \varepsilon(\mu|_{\cup_{k \in K}
    P_k}),\\
&= \sum_{\emptyset \neq K \subseteq \{1,\dotsc,d\}}
(-1)^{1+\sum_{k \in K} |P_k|} \prod_{k \in K} (-1)^{1 + |P_k|},\\
&= \sum_{\emptyset \neq K \subseteq \{1,\dotsc,d\}}
(-1)^{|K| + 1 + 2\sum_{k \in K} |P_k|},\\
&= \sum_{i=1}^d (-1)^{i+1} \binom{d}{i} = 1 - (1-1)^d = 1. 
\end{align*}
\end{proof}
Returning to the computation of the coefficient $c(w)$ in Equation~\ref{equation.cu},
let $P$ be the largest subset of $Q$ on which $w$ acts as a
permutation.
The coefficient $c(w)$ is by definition of $(D_{n,j})_{\otimes \ell}$,
\begin{align*}
c(w) &=  \sum_{R \in \P} (-1)^{|R|+1} \varepsilon(\mu_w|_R) = 1.
\end{align*}
Hence 
\begin{equation}
\sum_{\ell=1}^{|Q|} (-1)^{\ell+1} \trace((D_{n,j})_{\otimes\ell})
= \sum_{w \in E_{n,j}} w, \label{equation.cu2}.
\end{equation}
We get
\begin{align*}
\pi \P_n(\X_H)  
&=  \sum_{j=1}^n j E_{n,j},\\
&=  \sum_{\ell=1}^{|Q|} (-1)^{\ell+1} \sum_{j=1}^n j
\trace( \pi (D_{n,j})_{\otimes\ell}),\\
&=  \sum_{\ell=1}^{|Q|} (-1)^{\ell+1} 
\trace \sum_{j=1}^n j \sum_{k \geq 0} \pi
([H]_j[H^k]_{n-j})_{\otimes\ell},\\
&=  \sum_{\ell=1}^{|Q|} (-1)^{\ell+1} 
\trace \sum_{j=1}^n j \sum_{k \geq 0}
([\pi H _{\otimes\ell}]_j[ \pi H _{\otimes\ell}^k]_{n-j}),\\
&=  \sum_{\ell=1}^{|Q|} (-1)^{\ell+1} 
\trace \sum_{j=1}^n j  
[\pi H_{\otimes\ell} ]_j [(1-\pi H_{\otimes\ell})^{-1}]_{n-j}.
\end{align*}
\end{proof}

\begin{proposition}  \label{proposition.zetaXH}
Let $\A$ be a left reduced Dyck automaton. 
The zeta function of $\X_H$, where $H$ is one of the matrices
$C,C{M_c}^*,M_c$ defined from $\A$, is
\begin{equation*}
 \zeta_{\X_{H}}(z)  = \prod_{\ell=1}^{|Q|}
\det(I-H_{\otimes\ell}(z))^{(-1)^{\ell}},
\end{equation*}
where $H_{\otimes\ell}(z) = \theta\pi H_{\otimes\ell}$.
\end{proposition}
The same formula holds for $\X_H$ when $H$ is equal to $C$, ${M_r}^*C$ or
$M_r$ when $\A$ be a right reduced.
\begin{proof}

We get from Proposition~\ref{proposition.trace}
\begin{align*}
&\sum_{n \geq 1} 
\frac{\theta \pi P_n(\X_H)
}{n}\\
&=\sum_{n \geq 1} 
\frac{1}{n} 
\sum_{\ell=1}^{|Q|}  (-1)^{\ell+1}
 \trace\sum_{j=1}^{n} j [\theta\pi H_{\otimes\ell}]_j [(I -
 \theta\pi H_{\otimes\ell})^{-1}]_{n-j},\\
&= 
\sum_{\ell=1}^{|Q|} (-1)^{\ell+1}
\sum_{n \geq 1} 
\frac{1}{n} 
 \trace\sum_{j=0}^{n-1} (j+1) [\theta\pi H_{\otimes\ell}]_{j+1} [(I -
 \theta\pi H_{\otimes\ell})^{-1}]_{n-j-1},\\
&=  \sum_{\ell=1}^{|Q|} (-1)^{\ell+1}
\sum_{n \geq 1} 
\frac{1}{n} 
\trace \sum_{j=0}^{n-1} [\diff\theta\pi H_{\otimes\ell}]_{j} 
[(I - \theta\pi H_{\otimes\ell})^{-1}]_{n-j-1},\\
&=  \sum_{\ell=1}^{|Q|} (-1)^{\ell+1}
\sum_{n \geq 1} 
\frac{1}{n} 
\trace [ (\diff\theta\pi H_{\otimes\ell}) (I - \theta\pi H_{\otimes\ell})^{-1}]_{n-1},\\
&= \sum_{\ell=1}^{|Q|} (-1)^{\ell+1}
\sum_{n \geq 1} 
\frac{1}{n} 
\trace [ (\diff \log(I - \theta\pi H_{\otimes\ell})]_{n-1},\\
&=  \sum_{\ell=1}^{|Q|} (-1)^{\ell+1}
\sum_{n \geq 1} 
\trace [- \log(I - \theta\pi H_{\otimes\ell})]_{n},\\
&= \sum_{\ell=1}^{|Q|} (-1)^{\ell+1} \trace(- \log(I - \theta\pi H_{\otimes\ell})),
\end{align*}
where $\diff$ denotes the derivative with respect to the variable $z$.

Thus, using Jacobi's formula, we obtain 
\begin{align*}
\zeta(\X_H)(z)& = \exp 
\trace(\sum_{\ell=1}^{|Q|} (-1)^{\ell+1} (- \log(I - \theta\pi H_{\otimes\ell}))),\\
&= 
\det \exp(\sum_{\ell=1}^{|Q|} (-1)^{\ell}\log(I - \theta\pi H_{\otimes\ell}))),\\
&= \prod_{\ell=1}^{|Q|} \det (I - \theta\pi
H_{\otimes\ell})^{(-1)^{\ell}}.
\end{align*}
\end{proof}

\subsection{Computation of the zeta function of $X$}

The previous computations allow us to obtain 
directly the following general formula
for the zeta function of a sofic-Dyck shift $X$.
\begin{theorem} \label{theorem.zeta}
The zeta function of a sofic-Dyck shift accepted by a left reduced
Dyck (\resp right reduced) Dyck automaton
$\A$ (\resp $\B$) is given by the following formula, where $C$, $M_c$
and $C{M_c}^*$ are defined from $\A$ and 
$M_r$
and ${M_r}^*C$ are defined from $\B$.
\begin{align*}
 \zeta_{X}(z)  &= \frac{
   \zeta_{\X_{C{M_c}^*}}(z) 
   \zeta_{\X_{{M_r}^*C}}(z)
   \zeta_{\X_{M_c}}(z)  \zeta_{\X_{M_r}}(z)} 
{\zeta_{\X_C}(z)}\\
&= \prod_{\ell=1}^{|Q|} 
   \det(I-  (C{M_c}^*)_{\otimes\ell}(z))^{(-1)^{\ell}}
  \det(I-({M_r}^*C)_{\otimes\ell}(z))^{(-1)^{\ell}} \\
& \det(I-C_{\otimes\ell}(z))^{(-1)^{\ell} +1}
\det(I-M_{r,\otimes\ell}(z))^{(-1)^{\ell}}
\det(I-M_{c,\otimes\ell}(z))^{(-1)^{\ell}}.
\end{align*}
\end{theorem}
\begin{corollary} \label{corollary.zeta}
The zeta function of a sofic-Dyck shift is $\Z$-algebraic.
\end{corollary}


\begin{example}\label{example.zeta}
Let $\A$ be the Dyck automaton over $A$ pictured on the left part of 
Figure~\ref{figure.code}, where
$A=(\{a,a'\},\{b,b'\},\{i\})$.
The Dyck automaton $\A_{\otimes
  1}$ is the same as $\A$. The Dyck automaton $\A_{\otimes 2}$ is pictured
on the right part of Figure~\ref{figure.code}.
\begin{figure}[htbp]
    \centering
\begin{tikzpicture}[scale=0.4]
\node[state]            (3)         at   (14,0)                 {1};
\path     
(3) edge [in=35,out=335, loop] node {$-i$}  (3);
\node[state]            (1)         at   (0,0)                 {1};
\node[state]            (2) [right =2.5cm of 1]            {2};
\path     
(1)  edge[bend right] node {$i$} (2)
(2)  edge[bend right] node {$i$} (1)
(1) edge [in=100,out=40, loop] node  [inner sep=0.5mm,pos=0.5]   (a) {$b$} (1)
(1) edge [in=180,out=120, loop] node  [inner sep=0.5mm,pos=0.5]   (b) {$a$} (1)
(1) edge [in=250,out=190, loop] node  [inner sep=0.5mm,pos=0.5]   (c) {$a'$} (1)
(1) edge [in=330, out=270, loop] node  [inner sep=0.5mm,pos=0.5]   (d) {$b'$} (1)
(a) edge[matched edges, bend right] node {} (b)
(c) edge[matched edges, bend right] node {} (d);
\end{tikzpicture}
\caption{A Dyck automaton $\A$ (on the left) over 
$A=A_c \sqcup A_r \sqcup A_i$ with $A_c= \{a,a'\}$, $A_r= \{b,b'
\}$ and $A_i=\{i\}$ and the Dyck automaton
$\A_{\otimes 2}$ (on the right).
Matched edges are linked with a dotted
line.}
\label{figure.code}
\end{figure}
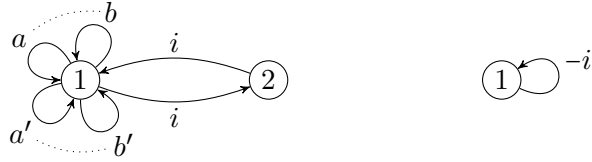
Let us compute the zeta function of $\X_\C$ for this automaton. 
Let
\[
C = \begin{bmatrix}
C_{11} & C_{12}\\
C_{21} & C_{22}
\end{bmatrix}, \:
C_{\otimes 2} = \begin{bmatrix}
C_{(1,2),(1,2)}
\end{bmatrix}.
\]
We have $C_{11}=  aD_{11}b + a'D_{11}b'$, 
$C_{22}= 0$, $C_{12}= i$, $C_{21}= i$,
with $D_{11}= aD_{11}bD_{11} + a'D_{11}b' D_{11} 
+ i i  D_{11} + \varepsilon$.
Hence 
\begin{equation*}
2z^2D^2_{11}(z) -(1-z^2) D_{11}(z) + 1 = 0
\end{equation*}
Since the coefficient of $z^0$ in $D_{11}(z)$ is $1$, we get
\begin{equation*}
D_{11}(z) = \frac{1 - z^2 - \sqrt{1- 10z^2 + z^4}}{4z^2}.
\end{equation*}
Hence 
\begin{equation*}
C_{11}(z) = 2z^2D_{11}(z)=\frac{1 - z^2 - \sqrt{1- 10z^2 + z^4}}{2}.
\end{equation*}
We have
$C_{22}(z)=0$, $C_{12}(z)=C_{21}(z)=z$.
We also have $C_{(1,2),(1,2)}=-i$ and thus $C_{(1,2),(1,2)}(z)=-z$.
Thus
\begin{align*}
\zeta_{X_{C}}(z) &=  \prod_{\ell=1}^{2}
\det(I- C_{\otimes\ell}(z))^{(-1)^{\ell}}\\
&= (1+z) 
\begin{vmatrix}
1- \frac{1-z^2 - \sqrt{1- 10z^2 + z^4}}{2}& -z\\
 -z&  1
\end{vmatrix}^{-1},\\
&=\frac{1+z}{1-z^2-\frac{1-z^2-\sqrt{1- 10z^2 + z^4}}{2}}.
\end{align*}
For $H=M_c,M_r$, we have 
\begin{align*}
\prod_{\ell=1}^{2} \det(I-H_{\otimes\ell}(z))^{(-1)^{\ell}} 
&= \frac{1}{1-2z}.
\end{align*}
We also have
\begin{align*}
C{M_c}^* &= 
\begin{bmatrix}
C_{11}& i\\
i & 0
\end{bmatrix}
\begin{bmatrix}
\{ a , a'\}^*& 0\\
0 & \varepsilon
\end{bmatrix}
= \begin{bmatrix}
C_{11}\{a, a' \}^*& i\\
i \{a , a' \}^*& 0
\end{bmatrix},\\
{M_r}^*C &= 
\begin{bmatrix}
\{ b , b' \}^*& 0\\
0 & \varepsilon
\end{bmatrix}
\begin{bmatrix}
C_{11}& i\\
i & 0
\end{bmatrix}
= \begin{bmatrix}
\{b , b' \}^* C_{11}& \{ b , b' \}^*i\\
i & 0
\end{bmatrix}.
\end{align*}
\begin{align*}
\prod_{\ell=1}^{2} 
\det(I- (C{M_c}^*)_{\otimes\ell}(z))^{(-1)^{\ell}} &= 
(1+z) 
\begin{vmatrix}
1- \frac{C_{11}(z)}{(1-2z)}& -z\\
-\frac{z}{1-2z}&  1
\end{vmatrix}^{-1}\\
&= 
\frac{(1+z)(1-2z)}
{
1-2z-z^2- C_{11}(z)
}.
\end{align*}
The same equality holds for ${M_r}^*C$.
We finally get
\begin{align*}
 \zeta_{X}(z)  &= 
\frac{(1+z) (1-z^2- C_{11}(z))}
{
(1-2z-z^2- C_{11}(z))^2
},\\
&= \frac{(1+z) (1-z^2- \frac{1 - z^2 - \sqrt{1- 10z^2 + z^4}}{2})}
{
(1-2z-z^2-\frac{1 - z^2 - \sqrt{1- 10z^2 + z^4}}{2})^2
}.
\end{align*}
\end{example}
The above formula shows that the zeta function of a sofic-Dyck shift
is a $\Z$-algebraic series.
It is proved in \cite{BealBlockeletDima2014b} that 
the zeta function of a finite-type-Dyck shifts  is the
generating series of an unambiguous context-free language, \ie
is an $\N$-algebraic
function. We conjecture that the result also holds for
sofic-Dyck shifts.

There is no known criterion for a $\Z$-algebraic series with coefficients in $\N$
to be $\N$-algebraic but there are some
necessary conditions on the asymptotic behavior of the
coefficients (see the Drmota-Lalley-Woods Theorem
in \cite[VII.6.1]{FlajoletSedgewick2009} and recent insights from
Banderier and Drmota in
\cite{BanderierDrmota2013, BanderierDrmota2013b}).

\noindent \textbf{Acknowledgements} The authors would like to thank
Arnaud Carayol for pointing us some mistakes in a preliminary version
of this paper and Pavel Heller and Wolfgang Krieger for helpful comments.
\begin{small}
\bibliographystyle{abbrv} 
\bibliography{Dyck}
\end{small}

\end{document}